\documentclass[11pt]{article}

\usepackage[a4paper,margin=2.7cm]{geometry}
\usepackage{graphicx}
 \usepackage{multirow} 
\usepackage{natbib}
\usepackage{amsthm}
\usepackage{amsmath}
\usepackage{bbold}
\usepackage{booktabs}
\usepackage{makecell}
\usepackage{paralist}
\usepackage{xcolor}
\usepackage{colortbl}
\usepackage{steinmetz}

\newtheorem{proposition}{Proposition}
\usepackage{anyfontsize} %
\usepackage{colortbl}   
\usepackage{xcolor}     
\usepackage{steinmetz}

\usepackage{natbib}
\usepackage{bbold}
\usepackage{booktabs}
\usepackage{makecell}
\usepackage{amsmath}
\usepackage{paralist}

\DeclareMathAlphabet{\mathpzc}{OT1}{pzc}{m}{it}

\newcommand{\LR}{\textrm{LR}}

\graphicspath{{./images}}

\title{A Decision-Theoretic Framework for Comparing Likelihood Ratio Methods for the Rare Type Match Problem}

\author{
Giulia Cereda\thanks{Department of Statistics, Computer Science, Applications (DiSIA), University of Florence, Italy}
\and
Fabio Corradi\footnotemark[1]
\and
Cecilia Viscardi\thanks{Department of Economics and Statistics (DISES), University of Salerno, Italy. Corresponding author: cviscardi@unisa.it}
}

\date{}

\begin{document}
\maketitle

\begin{abstract}The rare type match problem is a challenging situation faced by a forensic statistician who aims at  providing the value of a match between some characteristic of a crime stain and the corresponding characteristic of a suspect's stain when this characteristic has not been observed before.
Several methods have been designed in the literature to assess likelihood ratios for the rare type match case when evidence consists of a  Y-STR profile found on the crime scene matching the Y-STR profile of a designated suspect. 
We develop a general Bayesian decision-theoretic framework for quantifying the expected cost of alternative approaches using the logarithmic scoring rule. The framework  provides a novel formalization of the posterior cross-entropy, which explicitly enhances the contribution of the method-specific strategy used to extract information from the data. We revisit existing decompositions of posterior cross-entropy and introduce a new complementary decomposition that provides a more interpretable characterization of this contribution. This work compares nine different methods by assessing their performance using  empirical validation experiments. Its ultimate goal is to provide forensic experts and triers of fact with methodological insight and a pre-experimental guide to these alternative approaches.
\end{abstract}


\maketitle


\section{Introduction}\label{sec:intro}

  In forensic DNA analysis, a question is often whether a DNA stain found at a crime scene originates from an identified source, typically the suspect.  
Although the DNA profiles from the two sources may be identical, such correspondence does not constitute conclusive proof of identity. An alternative hypothesis, generally supported by the defence, is that the stain may have been left by another individual with an identical profile.

The \emph{likelihood ratio} (LR) is defined as the ratio of the probabilities of observing the evidence under the prosecution and defence hypotheses. 
The computation of the LR requires an assessment of the rarity of the observed DNA profile in a relevant population of possible alternative sources that in principle should be clearly specified. This becomes challenging when a database search reveals that the matching profile has not been previously observed — the so-called \emph{rare type match problem} \citep{cereda:2015}, also referred to as the fundamental problem of forensic statistics \citep{brenner:2010}. This situation frequently arises with Y-chromosome Short Tandem Repeat (Y-STR) haplotypes, where the number of possible haplotypes is extremely large leading to sparse database coverage. 

Several methods have been proposed to evaluate the LR associated with a  Y-STR match, see \citet{andersen:2021} for a review.  
This work focuses on a subset of those methods that can be applied in the rare type match context, including the \emph{$\kappa$-method} \citep{brenner:2010, brenner:2014}, the \emph{generalised Good method} \citep{cereda:2015b}, the \emph{two-parameter Poisson–Dirichlet method} \citep{cereda:2019, cereda:2022}, the \emph{augmented count} and its \emph{Bayesian} version \citep{cereda:2015}, the \emph{Discrete Laplace model} considering one or more than one subpopulations \citep{andersen:2013b},   and the \emph{Graphical Model  approach} \citep{andersen:2018}.  
The first three of the list were specifically developed for rare type matches, while the others are  more generally applicable.  
We first present a systematic overview of these methods and classify them according to the type of information they extract from the data, namely their method-specific summaries. 

Within forensic science, the LR is widely accepted as the logically coherent framework for expressing the strength of evidence \citep{robertson:2016, aitken:2004, buckleton:2020, taroni:2021}. It is  used across Europe— where is embedded in evaluative reporting frameworks such as those promoted by European Network of Forensic Science Institutes  \citep{enfsi:2015} and the UK Forensic Science Regulator \citep{fsr:2023}—as well as in Australia and New Zealand \citep{catoggio:2019} and is also employed in the United States, particularly in forensic DNA interpretation  \citep{swgdam:2018}. More generally, recent reviews of forensic science have stressed the need for empirical validation studies and quantitative assessment of the performance and error rates of forensic methods before their use in court \citep{pcast:2016, nrc:2009}. This requirement motivates the development of principled frameworks for evaluating and comparing alternative likelihood-ratio based methods such as the one we present here in which we provide both methodological and empirical contributions to the evaluation of  methods for rare type match problems based on Bayesian decision-theoretic framework. 
The framework relies on posterior cross-entropy \citep{cover:2012}, a measure of the expected logarithmic cost originally introduced to forensic inference by \cite{brummer:2006} and only rarely applied in the forensic DNA context \citep{vanlierop:2024}. We introduce a novel formulation that explicitly characterizes the contribution of the method-specific summaries employed by each method to posterior cross-entropy and its components.

We also revisit existing decompositions of posterior cross-entropy, based on method-specific summaries \citep{brocker:2009, degroot:1982}, focusing on the notions of \emph{potential score}, \emph{reliability}, and \emph{resolution}. We further propose a complementary decomposition that explicitly isolates a new component, named  \emph{coarseness}, thereby providing additional insight into the role of the information retained by the method-specific summaries.

Finally, we comparatively assess the performance of the nine methods through an extensive simulation study, using different population's  sizes and rates of success in reproduction. The empirical analysis provides insights into the strengths and limitations of the different approaches, contributing to a more informed comparison of alternative methods from a pre-experimental perspective.

The paper is organized as follows: Section~\ref{sec:rare} formalizes  the rare type match problem and summarizes the methods we compare; Section~\ref{sec:evaluation} gives account of some results belonging to information and decision theory, defines the evaluation framework and introduces the novel decomposition. Section~\ref{sec:empposcross}  provides   the empirical counterpart of such tools. Details on simulations required for the empirical evaluation of the methods are in  Section~\ref{sec<:details}, followed, in Section~\ref{sec:results}, by the main results we achieved in the experiments. Conclusions and directions for future research are in Section~\ref{sec:conclusions}.

\paragraph*{Notation} -- 
Throughout the paper, we use the following notation.
Random variables and their realizations are denoted by uppercase and lowercase letters, respectively.
With a slight abuse of notation, $p(\cdot)$ is used for discrete variables, while $\Pr(\cdot)$ refers to events.

Unless otherwise stated, $\log$ denotes the base-2 logarithm.
$H(X)$ is the entropy of $X$, $D(\cdot; \cdot)$ the Kullback–Leibler divergence between two distributions, and the $\mathcal{P}$ operator stands for the expectation  with respect to the data distribution.

\section{The rare type match problem for Y-STR profiles and proposals for LR evaluation}\label{sec:rare}
Y-STRs are polymorphic loci on the Y chromosome, the sex chromosomes transmitted exclusively from father to male offspring \citep{butler:2010}. Each Y-STR locus contains a short repeated sequence of nucleotides (e.g., AGGT), and individuals differ according to the number of repetitions (e.g., AGGT–AGGT–AGGT). The number of repeats at each locus defines an   \emph{allele}, while a Y-STR \emph{haplotype} is the combination of alleles across several carefully chosen loci.

The widespread method to type DNA profiles uses loci on non-sexual chromosomes, called autosomal, but Y-STR haplotypes provide a valuable tool to identify male-specific DNA in mixtures, such as vaginal swabs after sexual assaults, where autosomal STR profiles are often dominated by the victim’s DNA \citep{gill:2001}. Using male-specific markers such as Y-STR overcomes this limitation and enables targeted identification of a possible assailant.

Two important aspects of Y-STR inheritance are worth emphasizing. First, all   males belonging to the same patrilineal lineage share the same Y-STR haplotype unless a mutation occurs, typically through the gain or loss of a single repeat unit, and this is due to the fact that Y-STR haplotypes are inherited from father to son, as a non-recombining block. From this, it follows that loci are not independent and the probability of a haplotype cannot be derived by simply multiplying allele probabilities \citep{caliebe:2015}.  
 
As a result, because databases capture only a small fraction of the full space of possible Y-STR haplotypes, estimating haplotype frequencies directly from database counts is problematic, and the rare type match problem is particularly common in this context.

\subsection{Formal definition of the problem}
We consider the simplest case of a single male contributor to a 
mixed biological sample, where the use of Y-STR markers is motivated 
by the predominance of female DNA that precludes reliable autosomal 
profiling of the male component. The recovered Y-STR profile matches 
that of a known suspect, and we assume that there are no laboratory errors, and only one male contributes to the mixture. We will refer to the mixture as "stain" in the forthcoming sections.  

Let $H \in \{h_p, h_d\}$ denote the unobserved variable representing the two hypotheses: the prosecution one, ``the source of the male DNA is the suspect'' ($h_p$), and the defence one, ``the source of the male DNA is a random man from the relevant population'' ($h_d$) \citep{bright:2024}. 

The available data can be expressed as a realization $y_{1:n+2}$ of the random variables $Y := Y_{1:n+2}$, where:
\begin{itemize}
    \item $y_{1:n}$ is the reference database of $n$ Y-STR profiles from the relevant population;
    \item $y_{n+1}$ is the Y-STR profile from a stain recovered from the crime scene;
    \item $y_{n+2}$ is the Y-STR profile of the suspect;
    \item $y = \{y_{1:n}, y_{n+1}, y_{n+2}\}$ denotes the full set of available data.
\end{itemize}

We emphasize that the reference database is treated as part of the evidence to be evaluated, in line with a rigorous full Bayesian approach.
According to this notation, the likelihood ratio is expressed as
\[
\LR = \frac{p(y \mid h_p)}{p(y \mid h_d)}.
\]

In the rare type match problem, the observed pattern is such that $y_{n+1} \neq y_i \forall i \in \{1,\ldots,n\}$, and $y_{n+1} = y_{n+2}$. We denote $\tilde{y} := y_{n+1} = y_{n+2}$ as the rare haplotype.

According to the distinction proposed by \cite{Vergeer23}, the present setting belongs to the class of trace-reference problems for which a representative sample from the population of potential sources is available. While the propositions of interest are specific-source propositions, the existence of a population database makes it possible to model the suspect profile as a random realization from the same population represented by the database. We adopt this perspective throughout the paper. Therefore, although the inferential problem is specific-source in nature, the probabilistic treatment of the suspect profile is closer to a common-source formulation, where uncertainty about the source characteristics is explicitly represented rather than conditioned upon as fixed and known.

\subsection{A selection of proposals}
\label{sec:selection}

This paper compares a set of methods that can be used to compute the  LR in the rare type match problem, when the most straightforward approach—using the reciprocal of the relative frequency $n_{\tilde{y}}$ of the matching profile in the reference database is infeasible, since $n_{\tilde{y}}=0$.

Each method, denoted by $m$, is associated with a statistical model that defines a probability distribution $p_m(y,h)$.
Accordingly, each method produces an estimate of the LR based on its own summary statistics, denoted by $s_m(y)$.  
We write $\widehat{\LR}_m$ for the likelihood ratio obtained by a frequentist estimation of the model parameters.  
For Bayesian methods, the hat notation, $\widehat{\cdot}$, is omitted since parameters are integrated out.

Below, we briefly describe each method.

\begin{description}
\item[Augmented Count method(AC)] — first introduced by \cite{balding:2005}.  
The method assumes that $y_i $ are i.i.d. and follow a categorical distribution indexed by the parameter vector $\phi$, whose elements represent the population probabilities of the  haplotypes.  
The corresponding likelihood ratio is
\begin{align}
\label{eq:LR}
\LR_{\text{AC}} 
&= \frac{p(y \mid h_p)}{p(y \mid h_d)} 
= \frac{p_{\text{AC}}(y_{n+2} \mid y_{1:n+1}, h_p)p_{\text{AC}}(y_{1:n+1} \mid h_p)}{p_{\text{AC}}(y_{n+2} \mid y_{1:n+1}, h_d)p_{\text{AC}}(y_{1:n+1} \mid h_d)} \notag \\
&= \frac{1}{p_{\text{AC}}(y_{n+2} \mid y_{1:n+1}, h_d)} 
= \frac{1}{\phi_{\tilde{y}}},
\end{align}
where the last equality follows from the assumption that $p(y_{n+2} \mid y_{1:n+1}, h_p) = 1$ in case of match, when measurement errors are neglected, and from the fact that the first $n+1$ observations are independent of the hypotheses.
The parameter $\phi_{\tilde{y}}$ is estimated using the relative frequency of the crime stain’s profile in $y_{1:n+1}$ (i.e., the database including the crime stain but excluding the suspect’s profile):
\[
\hat{\phi}_{\tilde{y}} = \frac{n_{\tilde{y}} + 1}{n + 1} = \frac{1}{n + 1},
\]
since $n_{\tilde{y}}=0$.  
This yields the frequentist estimate:
\begin{equation}
\label{eq:AC}
\widehat{\LR}_{\text{AC}} = \frac{1}{\hat{\phi}_{\tilde{y}}}=n + 1.
\end{equation}
Hence, for rare type matches, the AC method produces an LR entirely determined by the database size $n$ and it is insensitive to the actual content of the reference database and to the specific profile $\tilde{y}$.  
In such cases, the resulting LR is necessarily larger than that obtained when the matching haplotype is observed more than once in the database.

\item[Bayesian Augmented Count method (B-AC)] — proposed by \cite{cereda:2015}.  
This approach extends the AC method by adopting a Bayesian framework. It assumes that the number of occurrences of the matching haplotype $\tilde{y}$ in the database follows a Binomial distribution with parameter $\phi_{\tilde{y}}$, whose prior distribution is $\mathrm{Beta}(\alpha, \beta)$.  
Under the standard non informative choice $\alpha = \beta = 1$, and exploiting the conjugacy of the Beta–Binomial model, the likelihood ratio is given by
\begin{align}
\label{eq:BAC}
\LR_{\text{B-AC}} 
&= \frac{1}{p_{\text{B-AC}}(\tilde{y} \mid y_{1:n+1}, h_d)}  \nonumber\\
&= \frac{1}{\int_{0}^{1} p_{\text{B-AC}}(\tilde{y} \mid \phi_{\tilde{y}}, h_d)\, p_{\text{B-AC}}(\phi_{\tilde{y}} \mid y_{1:n+1})\, d\phi_{\tilde{y}}} \notag \\[4pt]  
&= \frac{n + \alpha + \beta + 1}{n_{\tilde{y}} + \alpha + 1}
= \frac{n + 3}{n_{\tilde{y}} + 2}
= \frac{n + 3}{2}. 
\end{align}
The first equality follows from \eqref{eq:LR} and $\LR_{\text{B-AC}}$ is always smaller than $\LR_{\text{AC}}$, making it a more conservative approach—an appealing property in forensic contexts.

\item[Two-parameter Poisson–Dirichlet method (2PD)] — proposed by \cite{cereda:2019, cereda:2022}.  
This method provides a nonparametric Bayesian solution to the rare type match problem.  
It assumes an infinite number of possible Y-STR profiles in the population and models the vector of their ordered relative frequencies as following a Poisson–Dirichlet (PD) distribution \citep{pitman:1992b} with parameters $(\alpha, \theta)$, where $\alpha \in (0,1)$ and $\theta > -\alpha$.
The data $y_{1:n+2}$ are reduced to a partition $\pi_{[n+2]}$ of the set $\{1, \ldots, n+2\}$, where each subset contains the indices of observations sharing the same Y-STR profile.  
The relevant statistic for LR computation is the partition of $n+1$, $\pi_{n+1} := (n_1, n_2, \dots, n_k)$, where $n_i$ denotes the number of classes of size $i$ and $k$ is the largest class size in $\pi_{[n+1]}$.
The Pitman sampling formula \citep{pitman:1995} gives the probability of a random partition $\Pi_n$ as a function of the PD parameters $(\alpha, \theta)$:
\begin{equation}\label{ChEPPtF}
\Pr(\Pi_{n} = \pi_{n} \mid \alpha, \theta)
= \frac{[\theta + \alpha]_{k-1; \alpha}}{[\theta + 1]_{n-1; 1}}
  \prod_{j=1}^{k} [1 - \alpha]_{n_j - 1; 1},
\end{equation}
where $[\cdot]_{k,\alpha}$ indicates the ascending factorial.
As derived by \cite{cereda:2022}, the likelihood ratio can   be expressed as
\begin{align}\label{df1}
\LR_{\text{2PD}}
&= \frac{p(y \mid h_p)}{p(y \mid h_d)}
 = \frac{p_{\text{2PD}}(\pi_{[n+2]} \mid h_p)}
        {p_{\text{2PD}}(\pi_{[n+2]} \mid h_d)}\nonumber \\ 
&= \frac{1}
        {\int_0^1 \int_{-\alpha}^{+\infty} \frac{1 - \alpha}{n + 1 + \theta}\,
               p(\alpha, \theta \mid \pi_{[n+1]})\, d\alpha\, d\theta}.
\end{align}
The LR can be evaluated in two ways: 

\begin{itemize}
\item[(a)] \textbf{Fully Bayesian approach.}
Posterior samples $\{(\alpha_j, \theta_j)\}_{j=1}^{J}$ are drawn from $p(\alpha, \theta \mid \pi_{[n+1]})$ using an Approximate Bayesian Computation algorithm, yielding the approximation
\[
\LR_{\text{2PD}} \approx 
\frac{J}{\sum_{j=1}^{J} \frac{1 - \alpha_j}{n + 1 + \theta_j}}.
\]

\item[(b)]  \textbf{Empirical approach.}
The maximum likelihood estimates $(\widehat{\alpha}_{\mathrm{MLE}}, \widehat{\theta}_{\mathrm{MLE}})$ are substituted directly into \eqref{df1}, giving
\begin{equation}\label{eq:PDMLE}
\widehat{\LR}_{\text{2PD}} 
= \frac{n + 1 + \widehat{\theta}_{\mathrm{MLE}}}
       {1 - \widehat{\alpha}_{\mathrm{MLE}}},
\end{equation}
where the summary statistics in the likelihood for ($\alpha, \theta$) are $n$ and $\pi_{n+1}$.
\end{itemize}
\end{description}

The next two methods, in addition to disregarding the genetic information contained in the $n+2$ observed profiles, also ignore the specific profile labels, retaining only two relational events among them:
\begin{itemize}
\item $S$ (singleton): the crime stain profile $y_{n+1}$ is different from all values in $y_{1:n}$;
\item $M$ (match): the suspect’s profile matches the crime stain, i.e., $y_{n+2} = y_{n+1}$.
\end{itemize}

\begin{description}
\item[Generalised Good method (GG)] — introduced by \cite{cereda:2015b}.  
This nonparametric approach defines the LR as a ratio of the probabilities of the events $S$ and $M$ under the prosecution and defence hypotheses:
\begin{equation}\label{eq:lrgg}
\LR_{\text{GG}} = 
\frac{\Pr_{\text{GG}}(S, M \mid h_p)}{\Pr_{\text{GG}}(S, M \mid h_d)}
= \frac{\Pr_{\text{GG}}(S \mid h_p)}
        {\Pr_{\text{GG}}(S, M \mid h_d)}.
\end{equation}
The numerator of \eqref{eq:lrgg} is estimated using the Good–Turing estimator, $n_1 / n$, where $n_1$ is the number of singletons in $y_{1:n}$.  
The denominator, corresponding to the joint occurrence of $S$ and $M$, is estimated as $2n_2 / (n(n-1))$, where $n_2$ is the number of doublets.  
Approximating $n-1$ with $n$, we obtain
\[
\widehat{\LR}_{\text{GG}} = \frac{n\, n_1}{2 n_2},
\]
which is a quasi-unbiased estimator of $\LR_{\text{GG}}$, (see \citet{cereda:2015b} for details).  
The method thus relies on three summary statistics: $n_1$, $n_2$, and $n$.

\item [Brenner’s kappa method (BK)] — proposed by \cite{brenner:2010, brenner:2014}.  
This approach expresses the LR as
\[
\LR_{\text{BK}} =
\frac{\Pr_{\text{BK}}(S, M \mid h_p)}
     {\Pr_{\text{BK}}(S, M \mid h_d)}
= \frac{1}{\Pr_{\text{BK}}(M \mid S, h_d)},
\]
exploiting the independence of $S$ and $M$ and the fact that $\Pr(M \mid h_p) = 1$.  
Conditionally on $S$, the event $M$ is decomposed into the intersection of three sub-events:
\begin{description}
\item[$M_1$:] $y_{n+2}$ is not a singleton in $y_{1:n+1}$, estimated by $\Pr(M_1 |S, h_d) = 1 - k$, where $k = (n_1 + 1)/(n + 1)$ is the fraction of singletons in the augmented database;
\item[$M_2$:] $y_{n+2}$ matches a singleton in $y_{1:n+1}$, with $\Pr(M_2 |M_1, S, h_d) = k$;
\item[$M_3$:] $y_{n+2}$ matches exactly $y_{n+1}$, with $\Pr(M_3 | M_1, M_2, S, h_d) = 1/(n_1 + 1)$.
\end{description}
The product of these three probabilities gives
\[
\widehat{\LR}_{\text{BK}} =
\frac{n + 1}{1 - k}
= \frac{(n + 1)^2}{n - n_1}.
\]
The BK method therefore requires only the pair of summary statistics $(n, n_1)$, representing a further reduction of data compared with the GG and 2PD methods.

\item [Discrete Laplace method (DL and DL-mix)] — as described in \citep{ andersen:2018b}.  
Under a Fisher–Wright population model, \cite{caliebe:2010} showed that the allele distribution at each Y-STR locus can be approximated by a \emph{Discrete Laplace distribution}:
\[
f(d; m, r) = \frac{1 - r}{1 + r} \, r^{|d - m|},
\]
where $d$ is the number of repetitions (allele), $m$ and $r$ denote location and dispersion parameters.  
\cite{andersen:2013b} later provided empirical evidence that the population frequency of an entire haplotype can be approximated by the product of per-locus allele probabilities.
 Given a sample of haplotypes \(\{d^i\}_{i=1}^n\), where each $d^i$ is an haplotype with $k$ loci, the maximum likelihood estimates (MLE) of the Discrete Laplace parameters for each locus $j=1, ..., k$, are given by  
\[
\hat{m}_j = \text{median} \{d^i_j\}_{i=1}^n, \quad \hat{r_j} = g\left(sda_j\right),
\]  
where  $sda_j=\frac{1}{n} \sum_{i=1}^n |d^i_j - \hat{m}_j|$, is the mean absolute deviation from the median and $g(\cdot)$ is a known deterministic function. 
As a consequence, estimating the probability of observing a haplotype with \( k \) loci in the reference population requires a total of \( 2k \) summary statistics. Given the suspect Y-STR haplotype \( (d_1^s, \dots, d_k^s) \), the likelihood ratio (LR) is estimated as the reciprocal of the haplotype probability:  
\[
\widehat{\text{LR}}_{\text{DL}} = \frac{1}{\prod_{j=1}^k \frac{1-\hat{r}_j}{1+\hat{r}_j} \hat{r}_{j}^{{|d_j^s - \hat{m}_j|}}}.
\]  
An extension, called \emph{Disclapmix} (DL-mix), models haplotype frequencies as a mixture of Discrete Laplace components, each corresponding to a potential subpopulation.  
DL-mix also estimates the posterior probability that the suspect originates from each subpopulation. 


 \item [Graphical Model approach (GM)] — proposed by \cite{andersen:2018}.  
It is well known that, for Y-STR data, the transmission of haplotypes across generations induces dependencies among loci.  
Estimating their full joint distribution would require a prohibitive number of parameters and an extensive dataset.  
To address this, \cite{andersen:2018} proposed a simplified model in which dependencies are restricted to a tree structure.  
In this framework, each locus is conditionally dependent on at most one other locus, allowing the use of the Chow–Liu algorithm to identify the pairwise dependencies that maximize mutual information.  
For a set of $k$ loci, having estimated the parameters $\Pi$ of the corresponding conditional probability tables, the population probability of observing a haplotype identical to that of the suspect, with alleles $d = \{d_1, \dots, d_k\}$, is given by
\[
p_{\text{GM}}(d \mid \hat{\Pi}) = \prod_{j=1}^{k} p_{\text{GM}}(d_j \mid \mathrm{pa}(d_j), \hat{\Pi}),
\]
where $\mathrm{pa}(d_j)$ denotes the parent locus of $d_j$ in the tree.  
Accordingly, the estimated likelihood ratio is
\[
\widehat{\LR}_{\text{GM}} = 
\frac{1}{\prod_{j=1}^{k} p_{\text{GM}}(d_j \mid \mathrm{pa}(d_j), \hat{\Pi})}.
\]
To estimate $\hat{\Pi}$ tables of alleles counts are required, denoted $L_j$ in the forthcoming text.

\item[The empirical benchmark LR 
(EB)] — If the entire population were known, and under the i.i.d. categoric distribution assumption,  the exact likelihood ratio would be
\begin{equation}\label{eq:LRgold}
\LR = \frac{1}{\phi_{y_{n+1}}},
\end{equation}
where $\phi_{y_{n+1}}$ denotes the population relative frequency of haplotype $y_{n+1}$.  
In our simulation framework, this quantity will be available and used as a \emph{benchmark} for evaluating the different  methods.
 
\end{description}

\subsection{A classification of the proposals}
\label{sec:firstresult}

As a preliminary outcome of our analysis, we propose a classification of the reviewed methods into three groups (Table~\ref{tab:met}), according to how they exploit the information contained in the data.

\emph{Group 1} includes the Augmented Count (AC) method and its Bayesian version (B-AC).  
These two methods are the most extreme in terms of data reduction, using only the database size and the frequency of the matching haplotype.

\emph{Group 2} includes Brenner’s $\kappa$-method (BK), the Generalised Good (GG) method, and the two-parameter Poisson–Dirichlet (2PD) models.  
These approaches deliberately ignore the specific allelic composition of haplotypes and instead exploit the spectrum of haplotype counts (e.g., singletons, doublets, etc.).

\emph{Group 3} includes the Discrete Laplace (DL, DL-mix) and Graphical Model (GM) methods, which make full use of genetic information by modelling the allele distributions at each locus across all observed haplotypes, exploiting information about the specific allele numbers.

In contrast with Group~3, the methods in Groups~1 and~2 treat haplotypes as categorical labels without genetic meaning.  
Moreover, while Groups~1 and~3 aim to estimate the LR as the reciprocal of the haplotype frequency, the methods in Group~2 directly estimate or approximate the LR itself.  
This diversity of perspectives highlights the importance of developing a unified framework to compare their performance.
\begin{center}
\begin{table*}[ht!]
\caption{Classification of the reviewed methods according to their level of data reduction.}
\centering
\begin{tabular}{|p{1.5cm}|l|p{10.3cm}|}
\hline
\textbf{Group} & \textbf{Method} & \textbf{Description} \\
\hline
\multirow{2}{*}{Group 1}
 & AC & Augmented Count \\
 & B-AC & Bayesian Augmented Count\\
\hline
\multirow{4}{*}{Group 2} 
 & 2PD$\_$Emp & Two-parameter Poisson Dirichlet method with  MLE estimates   \\
 & 2PD  & Two-parameter Poisson Dirichlet method with ABC inference  \\
 & BK & Brenner’s $\kappa$-method \\
 & GG & Generalised Good method \\
\hline
\multirow{3}{*}{Group 3}  
 & DL & Discrete Laplace \\
 & DL-mix & Disclapmix allowing for subpopulations   \\
 & GM & Graphical Model assuming tree-structured dependences \\
\hline
\end{tabular}
\label{tab:met}
\end{table*}
\end{center}

In Table \ref{tab:es2} are the data concerning a toy example of a rare type match problem case, using  three loci haplotypes.
 In Table \ref{tab:es} we illustrate the relevant statistics for each methods. 
 \begin{center}
\begin{table*}[!ht]%
\centering
\small
\caption{The reference data base made of ten 3-loci Y-STR haplotypes for the toy example}
\label{tab:es2}

\begin{tabular}{c}
\toprule
$\begin{pmatrix}12\\13\\15\end{pmatrix}
 \begin{pmatrix}12\\11\\14\end{pmatrix}
 \begin{pmatrix}12\\11\\14\end{pmatrix}
 \begin{pmatrix}12\\11\\14\end{pmatrix}
 \begin{pmatrix}12\\13\\15\end{pmatrix}
 \begin{pmatrix}12\\11\\14\end{pmatrix}
 \begin{pmatrix}12\\11\\14\end{pmatrix}
 \begin{pmatrix}11\\10\\14\end{pmatrix}
 \begin{pmatrix}10\\12\\14\end{pmatrix}
 \begin{pmatrix}10\\12\\14\end{pmatrix}$\\
\bottomrule
\end{tabular}
\end{table*}
\end{center}
\begin{center}
\begin{table*}[ht!]
\centering
\small
\caption{Statistics required 
by each method to evaluate the LR. }
\label{tab:es}
\setlength{\extrarowheight}{4pt}
\begin{tabular}{ll}
\toprule
  \textbf{Method}  & \textbf{Statistics required} \\
\midrule
\multicolumn{2}{l}{\textit{Group 1 — count-based methods}} \\[2pt]
 AC              & $n=10,\ n_{\tilde{y}}=1$ \\
 B-AC     & $n=10,\ n_{\tilde{y}}=1$ \\
\midrule
\multicolumn{2}{l}{\textit{Group 2 — spectrum-based methods}} \\[2pt]
2PD
  & $n_1=1,\ n_2=2,\ n_3=0,\ n_4=0,\ n_5=1$ \\
 GG        
  & $n_1=1,\ n_2=2,\ n=10$ \\
BK       
  & $n_1=1,\ n=10$ \\
\midrule
\multicolumn{2}{l}{\textit{Group 3 — allele-frequency methods}} \\[2pt]
DL
  & $\hat{m}_1=12,\ \hat{m}_2=11,\ \hat{m}_3=14$;\quad
    $sda_1=0.5,\ sda_2=0.7,\ sda_3=0.2$ \\[6pt]
GM
  & $L_1=\begin{pmatrix}10&2\\11&1\\12&7\end{pmatrix}$\quad
    $L_2\mid L_1=\begin{pmatrix}
      & 10 & 11 & 12\\
      10 & 0  &  1 &  0\\
      11 & 0  &  0 &  5\\
      12 & 2  &  0 &  0\\
      13 & 0  &  0 &  2
    \end{pmatrix}$\quad
    $L_3\mid L_2=\begin{pmatrix}
      & 10 & 11 & 12 & 13\\
      14 & 1  &  5 &  2 &  0\\
      15 & 0  &  0 &  0 &  2
    \end{pmatrix}$ \\
\bottomrule
\end{tabular}
\end{table*}
\end{center}
\section{The evaluation of the proposals according to information theory}
\label{sec:evaluation}
In this section, we describe how tools developed within Bayesian decision theory \citep{degroot:2004} and information theory \citep{cover:2012} can be applied to evaluate and compare the methods introduced in Section~\ref{sec:selection}.  
Although these theories operate on probability distributions, while the methods here under consideration  output likelihood ratios, the two representations are directly related through
\[
p_{m}(h_p \mid y) = \frac{\LR_m \cdot \frac{p(h_p)}{p(h_d)}}{1 + \LR_m \cdot \frac{p(h_p)}{p(h_d)}}.
\]
  Since the posterior probability $p_m(h_p \mid y)$ depends on the prior odds we compare the methods across a common grid of prior probabilities $p(h_p)$, thereby avoiding dependence on any particular prior choice.

Within this framework, Bayesian decision theory provides a principled way to evaluate the quality of the posterior probabilities produced by each method. It does so through the use of \emph{scoring rules}, which quantify the cost associated with a reported posterior distribution $p_m(\cdot \mid y)$ for a specific case.
In particular, the logarithmic scoring rule proposed by \citet{good:1952} defines the cost of reporting $p_m(\cdot \mid y)$ when the true hypothesis is $h$ as
\begin{equation}\label{eq:log}
C_{\mathrm{log}}[p_m(\cdot\mid y); h] = -\log\big(p_m(h\mid y)\big).
\end{equation}

If  the posterior probability indicates with certainty the true hypothesis,  the cost is zero; 
otherwise, the cost is strictly positive and increases as $p_m(h_p\mid y)$ decreases.  

The logarithmic score is a strictly proper scoring rule, meaning that its expected value is uniquely minimized when the reported posterior distribution matches the true underlying one. Thus, it encourages honest reporting of beliefs rather than strategic distortion.  
 
This scoring rule imposes increasingly severe penalties as the assigned posterior probability for the true hypothesis decreases, diverging to infinity when that probability approaches zero. In the forensic context, where misleading certainty is particularly undesirable, this strictness makes the logarithmic score more suitable than other  alternatives, such as the quadratic score \citep{brier:1950}.
 
 The use of posterior cross-entropy as an evaluation criterion follows the decision-theoretic tradition based on strictly proper scoring rules \citep{good:1952, degroot:1982, brummer:2013, ramos:2013}. Similar ideas have recently been employed by \citet{vergeer:2021} to compare alternative classes of forensic LR systems. In the present work, the same principle is adopted for the rare type match problem.
 
 \subsection{Posterior cross-entropy}
\label{Sec:GeneEntropy} 
Equation~\eqref{eq:log} defines the cost associated with a reported posterior for a fixed $(h, y)$.  
To evaluate the overall performance of $p_m(h \mid y)$, Bayesian decision theory considers the expected cost with respect to both random variables $H$ and $Y$, taking into account the uncertainty about which hypothesis is true and about the possible observations.  
This leads  to a well-known information-theoretic measure, the \emph{posterior cross-entropy} ($\mathcal{P}CE$) adopted in this study:
\begin{align}
\mathcal{P}CE(p(h\mid y); p_m(h\mid y))
&= -\sum_{y \in \mathcal{Y}}\sum_{h\in\{h_p,h_d\}} p(h, y)\, \log p_m(h\mid y) \nonumber \\
&= -\sum_{h\in\{h_p,h_d\}} p(h)\sum_{y \in \mathcal{Y}} p(y \mid h)\, \log p_m(h\mid y) \label{xxd2}\\
&= -\sum_{y \in \mathcal{Y}} p(y)\sum_{h\in\{h_p,h_d\}} p(h \mid y)\, \log p_m(h\mid y), \label{xxd}
\end{align} that can be decomposed in several equivalent forms, each providing complementary insights into the different characteristics of the methods under comparison.
In \eqref{xxd2} the joint probability $p(h, y)$ is expressed as $p(h)p(y\mid h)$, this latter termed the \emph{reference distribution} for the data. In  \eqref{xxd} $p(h, y)$ is expressed as $p(y)p(h\mid y)$,  where $p(h\mid y)$ represents the corresponding \emph{reference posterior}.  
In practice, the reference distribution $p(y\mid h)$ is typically unknown, so it is  $p(h\mid y), $ unless we are  within a simulation framework.

\paragraph*{A reformulation of the classical decomposition —} From \eqref{xxd}, adding and subtracting the term $\log p(h\mid y)$, it follows that 
\begin{eqnarray}\label{eq:eqGoodrisk}
\mathcal{P}CE(p(h\mid y); p_m(h\mid y))=& \sum_{y \in \mathcal{Y}} p(y) \sum_{h\in\{h_p,h_d\}}p(h\mid y)\log\left(\frac{p(h\mid y)}{p_m(h\mid y)}\right)  
  +\sum_{y \in \mathcal{Y}}p(y)\sum_{h\in\{h_p,h_d\}}p(h\mid y)\log \frac{1}{p(h\mid y)}\nonumber\\
=&  \mathbb{E}_Y\big[D(p(h\mid  y);p_m(h\mid  y))\big]
 +\mathbb{E}_Y\big[H(p(h\mid y))\big].
\end{eqnarray}

The first component of  \eqref{eq:eqGoodrisk} is the expected value, with respect to $p(y)$, of the KL divergence between the reference posterior and the posterior produced by a generic method $m$.  
This decomposition highlights that methods yielding posterior distributions farther (in the KL sense)  from the reference posterior incur a higher expected cost.  
The second component corresponds to the expected entropy of the reference posterior, representing a baseline cost that cannot be reduced.

Let \( \mathcal{X}_m \) denote the set of distinct posterior values \( p_m(h_p \mid y) \) produced by method \( m \).  
For each \( x \in \mathcal{X}_m \), define
\[
\mathcal{Y}_m^x = \{\, y \in \mathcal{Y} : p_m(h_p \mid y) = x \,\},
\]
the subset of observations that yield the same posterior probability \( x \) for \( h_p \).  
The space \( \mathcal{Y} \) can be partitioned into the subsets \( \{ \mathcal{Y}_m^x \}_{x \in \mathcal{X}_m} \), and this  allows us to express \( \mathcal{P}CE \) by conditioning on the induced values \( x \) rather than on individual observations \( y \). Indeed, since each method \( m \) produces posterior probabilities based on  statistics \( s_m(y) \), it assigns the same value \( p_m(h_p \mid y) \) to all observations yielding the same \( s_m(y) \).

The transition from \( \mathcal{Y} \) to \( \mathcal{X}_m \) has four main consequences:
\begin{enumerate}
 \item The conditional distributions \( p(h \mid x) \), expressed in terms of the coarser variable \( x \),  correspond to averages of \( p(h \mid y) \) over the sets \( \mathcal{Y}_m^x \):
\[
p(h_p\mid x)=p(h_p \mid \mathcal{Y}_m^x)
=\frac{\sum_{y \in \mathcal{Y}_m^x} p(h_p \mid y)\, p(y)}
{\sum_{y \in \mathcal{Y}_m^x} p(y)}, \quad \forall x\in \mathcal{X}_m, 
\] and thus no longer retains the full granularity of the reference posterior across individual observations.  See Appendix~\ref{app:A1} for details.
This loss of detail is inherent to the method \( m \) and will be considered by decomposition \eqref{eq:ceci3}. 
    \item Although both \( p(h \mid y) \) and \( p(h \mid x) \) are generally unknown, the latter involves averages over subsets of \( \mathcal{Y} \), which makes it more amenable to empirical estimation, as will be illustrated in equation \eqref{eq:calemp} and in the subsequent discussion.
    \item This framework also allows for a formal definition of calibration.  
A method is said to be \emph{well-calibrated} with respect to a reference distribution \( p \) if
\begin{equation}
p_m(h_p \mid x) = p(h_p \mid x) = x, \quad \forall x \in \mathcal{X}_m.
\end{equation}
Thus, for each method \( m \) and each \( x \in \mathcal{X}_m \),  
\( p(h_p \mid x) \) represents the calibrated version of \( p_m(h_p \mid x) \), i.e., the probability that the reference posterior assigns when conditioning on the set \( \mathcal{Y}_m^x \).

\item Finally, under this formulation, equation~\eqref{eq:eqGoodrisk} can be rewritten as (see Appendix~\ref{app:A2}):
\begin{align}
\label{eq:2}
\mathcal{P}CE(p(h\mid y); p_m(h\mid y)) 
&= \mathbb{E}_{X_m}\big[D(p(h\mid  x); p_m(h\mid x))\big]
+\mathbb{E}_{X_m}\big[ H(p(h\mid  x))\big]\\
&= \text{Reliability loss} + \text{Potential score}. \nonumber
\end{align}
\end{enumerate}

The reliability loss quantifies miscalibration by measuring the expected KL divergence between the method’s posterior $p_m(h \mid x)$ and its well-calibrated counterpart $p(h \mid x)$.  
The potential score represents the expected cost associated with the well-calibrated version of $p_m$, reflecting the minimum achievable cost if the method were perfectly calibrated.

Below, we present two additional decompositions of the posterior cross entropy, based on two articulations of  the potential score    providing additional insights.

\paragraph*{The Br\"ocker decomposition —}

A further decomposition of the expected cross-entropy, first introduced by \cite{murphy:1973} for the Brier scoring rule, was extended to the logarithmic scoring rule by \cite{brocker:2009}. It takes the form:
\begin{align}
\label{eq:brocker}
\mathcal{P}CE(p(h\mid y); p_m(h\mid y))
&= \mathbb{E}_{X_m}\left[D\big(p(h\mid x);\, p_m(h\mid x)\big)\right]
  + H\big(p(h)\big)
  - \mathbb{E}_{X_m}\left[D\big(p(h\mid x);\, p(h)\big)\right] \\
&= \text{Reliability loss} + \text{Prior Uncertainty} - \text{Resolution}. \nonumber
\end{align}

This decomposition, that is derived in details in Appendix~\ref{app:A3}, separates the potential score into two components that depend on the prior over the hypotheses.  
The prior uncertainty is the entropy of the prior $p(h)$.  
The resolution is the expected KL divergence between the prior and the well-calibrated posterior; it quantifies the information gain achieved by updating the prior with the calibrated posterior of the method.
This ability depends on how the method partitions the data space $\{\mathcal{Y}_m^x\}_{x\in\mathcal{X}_m}$: partitions that allow posteriors to deviate (in KL sense) strongly from the prior yield to a higher resolution.

Resolution is closely related to discriminating power as explained in  \citet{brocker:2015}.  
A method such that $p_m(h\mid y)=p(h)$ for all $y$—is perfectly calibrated but non informative: its resolution and discriminating power are both zero (the partition collapses to the full space $\mathcal{Y}$).  
This trivial case underscores the need to combine good calibration with a high resolution.


\paragraph*{A novel decomposition —}

We propose an additional decomposition of the posterior cross-entropy (proof in Appendix~\ref{App:A4}):
\begin{align} 
\label{eq:ceci3}
\mathcal{P}CE(p(h\mid y); p_m(h\mid y))
&= \mathbb{E}_{X_m}\left[D\big(p(h\mid x);\, p_m(h\mid x)\big)\right]
  + \mathbb{E}_Y\left[H\big(p(h\mid y)\big)\right]+ \mathbb{E}_Y\left[D\big(p(h\mid y);\, p(h\mid x)\big)\right] \\
&= \text{Reliability loss} + \text{Posterior Uncertainty} + \text{Coarseness}. \nonumber
\end{align}

In Br\"ocker’s formulation \eqref{eq:brocker}, the potential score from \eqref{eq:2} was expressed by the prior entropy reduced by the information gained through the resolution.  
In contrast, our decomposition \eqref{eq:ceci3} views the same quantity as the minimal uncertainty about $H$—the expected entropy of the reference posterior, $\mathbb{E}_Y[H(p(h\mid y))]$—plus a penalty that accounts for the loss of detail introduced by the method’s specific data reduction.

The \emph{Coarseness} term is this penalty: it is the expected KL divergence between the reference posterior at each $y\in\mathcal{Y}$ and the reference posterior conditioned on the subset $\mathcal{Y}_m^x$ containing $y$.  
The coarseness quantifies the information loss due to data reduction and associated statistics, even under perfect calibration.

Finally note that, given two methods (1 and 2, say), it follows from decompositions~\eqref{eq:brocker} and~\eqref{eq:ceci3} that
\begin{equation}
\label{deltares}
\Delta(\text{Potential Score}_{1\text{-}2})= \Delta(\text{Coarseness}_{1\text{-}2}) =  -\Delta(\text{Resolution}_{1\text{-}2})
\end{equation}
so that two different interpretations of the difference in potential score between two methods are available.

\section{The empirical evaluation of the posterior cross entropy}
\label{sec:empposcross}

Up to this point, our discussion has focused on the theoretical formulation of the posterior cross-entropy and its decompositions, but 
since the reference distribution $p(y\mid h)$ is unknown, the posterior cross entropy and its decomposing terms must be estimated using an archive of  supervised cases,  
\[
\{(y^{(i)}, h^{(i)})\}_{i=1}^T,
\]
which are assumed to be distributed according to $p(y,h)=p(y\mid h)p(h)$.  
Here, $y^{(i)}$ denotes the observed data for case $i$, and $h^{(i)}$ the corresponding known true hypothesis.  

For each value of the prior $p(h_p)$ and method $m$, this allows us to compute the \emph{empirical posterior cross-entropy}, denoted $\mathcal{E}CE_m$, which is the Monte Carlo estimate of~\eqref{xxd2}:
\begin{align}
\label{eq:empece}
 \mathcal{E}CE_m 
&=-\frac{p(h_p)}{|C_p|}\sum_{i\in C_p} \log \bigl(p_{m}(h_p\mid y^{(i)})\bigr)
    -\frac{p(h_d)}{|C_d|} \sum_{i\in C_d}\log\bigl(p_{m}(h_d\mid y^{(i)})\bigr) \\[4pt]
\label{ECE0}
&= - \sum_{x\in \mathcal{X}_m}  \Bigl( p(h_p)\, \nu_m(x\mid h_p)\,  \log(x)
       + p(h_d)\, \nu_m(x\mid h_d)\,  \log(1-x) \Bigr)
\end{align}

where $C_p=\{i: h^{(i)}=h_p\}$ and $C_d=\{i: h^{(i)}=h_d\}$.  
Equation~\eqref{ECE0} is obtained by grouping cases according to the distinct posterior values in $\mathcal{X}_m$, with
\[
\nu_m(x\mid h_p) = \frac{\sum_{i=1}^T \mathbb{1}\{h^{(i)} = h_p, y^{(i)} \in \mathcal{Y}_m^x\}}{|C_p|}
\]
denoting the empirical relative frequency of posteriors equal to $x$ among prosecution cases, and analogously for $\nu_m(x\mid h_d)$.

The empirical counterparts of the decompositions in~\eqref{eq:2} and~\eqref{eq:brocker} are then given by:
\begin{align}
\label{eq:sl} 
\mathcal{E}RL_m &=
\sum_{x \in \mathcal{X}_m} \Bigl[
p(h_p)\, \nu_m(x \mid h_p)  \log\left(\frac{\hat{p}(h_p \mid x)}{x}\right) +
p(h_d)\, \nu_m(x \mid h_d)  \log\left(\frac{\hat{p}(h_d \mid x)}{1 - x}\right)
 \Bigr]\\
&=  \sum_{x \in \mathcal{X}_m} \nu_m(x)\, D(\hat{p}(h \mid x); x), \nonumber \\[4pt]
\label{eq:sl2}
\mathcal{E}PS_m &= - \sum_{x \in \mathcal{X}_m} \Bigl[
p(h_p)\, \nu_m(x\mid h_p)\,  \log\left(\hat{p}(h_p \mid x)\right) +
p(h_d)\, \nu_m(x\mid h_d)\,  \log\left(\hat{p}(h_d \mid x)\right)
\Bigr] \\
&= \sum_{x \in \mathcal{X}_m} \nu_m(x)\, H(\hat{p}(h \mid x)), \nonumber \\[4pt]
\label{eq:sl3}
\mathcal{E}R_m &= - \sum_{x \in \mathcal{X}_m} \Bigl[
p(h_p)\, \nu_m(x\mid h_p)\,  \log\left(\frac{\hat{p}(h_p\mid x)}{p(h_p)}\right)+
p(h_d)\, \nu_m(x\mid h_d)\,  \log\left(\frac{\hat{p}(h_d\mid x)}{p(h_d)}\right)
\Bigr] \\
&= \sum_{x \in \mathcal{X}_m} \nu_m(x)\, D(\hat{p}(h \mid x); p(h)). \nonumber
\end{align}

It holds that:

\begin{align}\label{emp_dec}
\mathcal{E}CE_m  &= \mathcal{E}RL_m+ \mathcal{E}PS_m \\
\label{emp_dec2} 
&= \mathcal{E}RL_m - \mathcal{E}R_m + H(p(h)) \nonumber
\end{align}

Here, $\mathcal{E}RL_m$, $\mathcal{E}PS_m$, and $\mathcal{E}R_m$ denote the empirical counterpart of the reliability loss, potential score, and resolution, respectively.
The reference distribution $p(h\mid x)$ represents the ideal, perfectly calibrated posterior.
In practice, this quantity is unknown and is estimated from data using a \emph{calibration method}, which provides the approximation $\hat{p}(h\mid x)$.

A natural empirical estimate is
\begin{equation} 
\label{eq:calemp}
\hat{p}(h\mid x) = \nu_m(h\mid x)
= \frac{\sum_{i=1}^T \mathbb{1}\{h^{(i)} = h,  y^{(i)} \in \mathcal{Y}_m^x\}}
{\sum_{i=1}^T  \mathbb{1}\{y^{(i)} \in \mathcal{Y}_m^x\}},
\quad h\in\{h_p,h_d\}.
\end{equation}
but, as noted by \citet{brocker:2012} and detailed in Appendix~\ref{App:H},
the estimators of the potential score and the reliability loss based on the empirical calibration~\eqref{eq:calemp} are affected by a bias of order $\pm |\mathcal{X}_m|/T$.
This bias arises because the number of simulated cases required to obtain stable estimates of $p(h\mid x)$ increases with the cardinality of $\mathcal{X}_m$:
the larger the partition induced by the method, the more cases are needed to stabilize the empirical frequencies. 

The limitations of the empirical estimator~\eqref{eq:calemp}, particularly its sensitivity to the number of available cases,
have motivated several authors to propose alternative approaches for estimating calibrated probabilities \citep{brummer:2010, ramos:2013}.
In binary classification, such probabilities can be estimated using either parametric or non-parametric methods, as discussed in \citep{filho:2023}.
Parametric methods involve fitting a model (e.g., logistic and beta models) that relates the predicted probabilities \(x\) to the empirical relative frequencies.  
In contrast, non-parametric approaches group the \(x\) values into bins according to specific criteria.  
By pooling multiple \(x\) values within each bin, these methods mitigate instability caused by a limited number of observations in the numerator of~\eqref{eq:pcalm}.  
Among them, the Pool Adjacent Violators (PAV) algorithm  (see Appendix \ref{sec:cal}) ensures that the estimated calibration function is monotonic, thereby preserving the natural ordering of the method’s posterior \(p_m\).   

The PAV algorithm has a particularly important property: as shown by \citet{ayer:1955} and \citet{brummer:2013}, it produces calibrated probabilities that minimize the empirical cross-entropy   
among all monotonic calibration methods.   
Consequently, the corresponding value of $\mathcal{ECE}$  approaches from above the cross-entropy corresponding to perfectly calibrated probabilities—namely, the empirical posterior entropy of \(p(h\mid y)\).  
In this sense, \(\hat{p}_m(h \mid y)\) represents the best attainable approximation of perfect calibration:  
when used to evaluate \(\mathcal{P}CE(\hat{p}_m(h \mid y))\), it provides the closest achievable (and slightly conservative) estimate of the potential score, while minimizing the reliability loss. Once the calibrated probabilities \(\hat{p}_m = p_{\mathrm{PAV},m}\) are obtained, two approaches can be followed.  
The first consists in plugging these values directly into equations~\eqref{eq:sl}–\eqref{eq:sl3}, thereby obtaining empirical estimates of the potential score and the reliability loss.  
Alternatively, one may compute one of these quantities by difference, exploiting the relationships among the decomposition terms.  
In both cases, the resulting estimates of the potential score and the reliability loss are systematically, respectively, over- and under-estimated, reflecting the conservative nature of the PAV calibration.

In the present framework, calibration plays a dual role: it improves reliability and, at the same time, provides the key ingredient required to estimate the decomposition terms of the posterior cross-entropy, especially the potential score.

 Looking at the coarseness term in the decomposition of~\eqref{eq:ceci3}, 
we note that it cannot be evaluated directly, not even empirically, 
since it depends on the reference posterior \(p(h\mid y)\), which—unlike \(p(h\mid x)\)—is neither observable nor estimable.  
However, it is still possible to provide a relative ranking of methods in terms of coarseness by means of pairwise comparisons with a fixed reference method, such as in equation \eqref{deltares}.

\section{Simulation study}
\label{sec<:details}
 Past real cases constitute the most genuine database of supervised cases, 
but they are hardly available—particularly in forensic identification—where the true origin of the traces often remains uncertain, even after a court decision.  

For this reason, we propose conducting a series of controlled simulation experiments based on several different reference populations, each reflecting different demographic scenarios.  
From each simulated population, we can draw numerous supervised cases of rare type matches.  
Each case consists of a sample of 10-loci Y-STR haplotypes from the population (the reference database), augmented by the haplotypes from both the crime scene and the suspect sample.  
The LRs obtained by applying the selected methods to each case are then evaluated and compared using the \emph{empirical posterior cross-entropy} (\(\mathcal{E}CE\)) and its associated decompositions—namely, the \emph{reliability loss}, the \emph{potential score}, the \emph{resolution} and the \emph{coarseness}.  
These quantities provide a principled framework for assessing both the calibration and the discriminative power of each method under controlled yet realistic population scenarios.

\subsection{Simulated population}
The simulation of the different populations is inspired to the method used in \citep{andersen:2017,andersen:2018b}, 
through the \texttt{Malan} package (version 1.0.4)  in \textsf{R}  \citep{andersen:2018b}.
Populations were simulated over 250 generations (approximately $5'000$ years) with a constant growth rate of 2\% per generation, reaching a final size in the last generation of $M \in \{5'000, 50'000, 500'000\}$ individuals.
The total population across the last three generations—those considered alive for forensic purposes—thus amounts to approximately \(N=3\times M\) individuals. See Table \ref{tab:population} for the values of $N$ and for the initial population size at generation 0.
These differences account for alternative specifications of the defence hypothesis, 
for instance whether the relevant population corresponds to an entire country or to a more restricted area, thereby influencing the specification of \(N\).
Simulations were also performed under three different assumptions on the variance in reproductive success (VRS), that combined with the three values of $M$ gives rise to
a total of nine simulated populations (see Table \ref{tab:population}).  

In population genetics, the variance in reproductive success (VRS) is defined as the excess variance in the number of sons per father, \(S_i\), relative to the Poisson expectation:
$\mathrm{VRS} := \mathrm{Var}(S_i) - 1.
$

Following \citet{andersen:2017}, we model the relative reproductive success of each individual \(i\) as an independent Gamma random variable with mean one:
\[
Q_i \sim \operatorname{Gamma}\left(\alpha,\frac{1}{\alpha}\right),
\]
where the second parameter is the scale parameter. Thus,
\[
\mathbb{E}[Q_i]=1,
\qquad
\operatorname{Var}(Q_i)=\frac{1}{\alpha}.
\]
Conditional on \(Q_i=q_i\), the number of sons fathered by individual \(i\), denoted \(S_i\), is modeled as
\[
S_i \mid Q_i=q_i \sim \operatorname{Poisson}(q_i).
\]
The resulting Gamma--Poisson mixture implies that \(S_i\) follows a negative binomial distribution, with
\[
\mathbb{E}[S_i]=1,
\qquad
\operatorname{Var}(S_i)=1+\frac{1}{\alpha}.
\]
Therefore,
\[
\mathrm{VRS}
:=\operatorname{Var}(S_i)-1
=\frac{1}{\alpha}.
\]

We considered three values of \(\alpha\), corresponding to increasing levels of reproductive success variability:
\begin{itemize}
    \item \(\alpha=\infty\) (\(\mathrm{VRS}=0\)): equal paternity probabilities for all males (the classical Wright-Fisher model);
    \item \(\alpha=2\) (\(\mathrm{VRS}=0.5\)): moderate variability in male reproductive success, as might be expected in many contemporary populations;
    \item \(\alpha=0.1\) (\(\mathrm{VRS}=10\)): high variability in male reproductive success, with paternity concentrated among a small number of highly prolific males.
\end{itemize}

In Table \ref{tab:population} we label the 9 different populations used for testing the reviewed methods.

\begin{table}[ht]
\caption{Names assigned to populations according to VRS and final population size $N=3\times M$, along with initial population size $IS$ derived by inverting the growth formula \(N_t = \lceil IS \times (1 + 0.02)^t \rceil\), \(t = 0, 1, \ldots, 249\), and reference database size $n= 0.02\times M$ .}
\centering

\begin{tabular}{lllccc}
\toprule
& & \multicolumn{4}{c}{VRS} \\
\cmidrule(lr){4-6}
$N$ & $IS$ &$n$& 0 & 0.5 & 10 \\
\midrule
$1'500'000$ & 35& 100000 &1A & 1B & 1C \\
$150'000$    &353 &10000&2A & 2B & 2C \\
$15'000$      & 3540&1000& 3A & 3B & 3C \\
\bottomrule
\end{tabular}

\label{tab:population}
\end{table}
These settings are not intended to be realistic representations of actual populations but are used to perform a sensitivity analysis on population assumptions.

To initiate each simulated genealogy, \emph{founder haplotypes} were assigned to the first generation.  
These haplotypes represent the initial genetic profiles from which the Y-STR lineages evolve through mutation and inheritance across generations.  
Founder profiles were sampled from the Dutch population in the YHRD database, 
and locus-specific mutation rates were applied to model the transmission process (Table~\ref{tab:ystr_mutation_rates}). 
  
\begin{center}

\begin{table*}[ht!]
\caption{Mutation rates for the Y-STR markers  from \citep{andersen:2017, willuweit:2015} , rounded to five decimal places.}
\label{tab:ystr_mutation_rates}
\centering
\begin{tabular}{cccccccccc}
\hline
DYS19 &
DYS389I &
DYS389II.I &
DYS390 &
DYS391 &
DYS392 &
DYS393 &
DYS437 &
DYS438 &
DYS439 \\
\hline

0.00224 &
0.00293 &
0.00119 &
0.00211 &
0.00245 &
0.00052 &
0.00105 &
0.00122 &
0.00037 &
0.00545 \\
\hline
\end{tabular}

\end{table*}
    
\end{center}

\subsection{The archive of simulated supervised cases}
\label{sec:archive}

After having created the 9 background populations, we  proceed by generating a collection of supervised rare type match cases.  
Each case consists of a reference database of Y-STR profiles, denoted \( y_{1:n} \), randomly sampled without replacement from the simulated population.  
The database size, $n$, is arbitrarily set to the \(2\% \) of the last generation of the corresponding reference population (see Table \ref{tab:population} for the specific values).  
To each database, we add two additional profiles: one representing the crime sample and one representing the suspect, as follows.

\begin{itemize}
\item Under the prosecution hypothesis \( h_p \), the crime sample profile \( y_{n+1} \) is drawn at random from the portion of the population conditional on its haplotype not being observed among \( y_{1:n} \), thereby ensuring the \emph{rare type} condition.  
The suspect profile \( y_{n+2} \) is then set equal to \( y_{n+1} \), since both originate from the same individual.

\item Under the defence hypothesis \( h_d \), the procedure is identical except that the suspect profile \( y_{n+2} \) is drawn independently from the entire population.
\end{itemize}

\section{Results}
\label{sec:results}

The results are organised  adopting two complementary perspectives:
\begin{itemize}
    \item  in Section \ref{sec:all-cases} we  compare the methods by considering all cases simulated by using the empirical estimates of the metrics introduced in Section~\ref{sec:evaluation}. 
    \item  
    in Section \ref{sec:match-cases} we  realize a conditional  comparison restricted to the match cases (also under $h_d$), where we focus on how the methods behave in terms of the distribution of the likelihood ratios when a match arises.
\end{itemize}

  Some of the results in Section \ref{sec:all-cases} are displayed for the population 2B, that is intermediate in terms of size and VRS (see Table \ref{tab:population}). This is done for space reasons since the results obtained  for all nine populations are qualitatively similar (see Supplementary material). We simulate (approximately) an equal number of match cases under both hypotheses.  
This requires generating a substantially larger number of cases under \( h_d \), due to the predominance of non-match events arising with probability one minus the random match probability, as shown in Table \ref{tab:Eterogeneity}.

\subsection{Comparison of the methods: all cases considered}
\label{sec:all-cases}

ECE plots \citep{ramos:2013} are commonly used to  display how empirical cross-entropy varies as a function of the prior probability assigned to one of the hypotheses.

In Fig.~\ref{fig:combined} we present the ECE together with its decomposition terms, according to \eqref{eq:eqGoodrisk}, for all the methods.  
For clarity, prior probabilities on the \( x \)-axis are reported on the \( \log_{10} \)-odds scale.

\begin{center}
\begin{figure*}[ht!]
    \centering
    \includegraphics[width=1\linewidth]{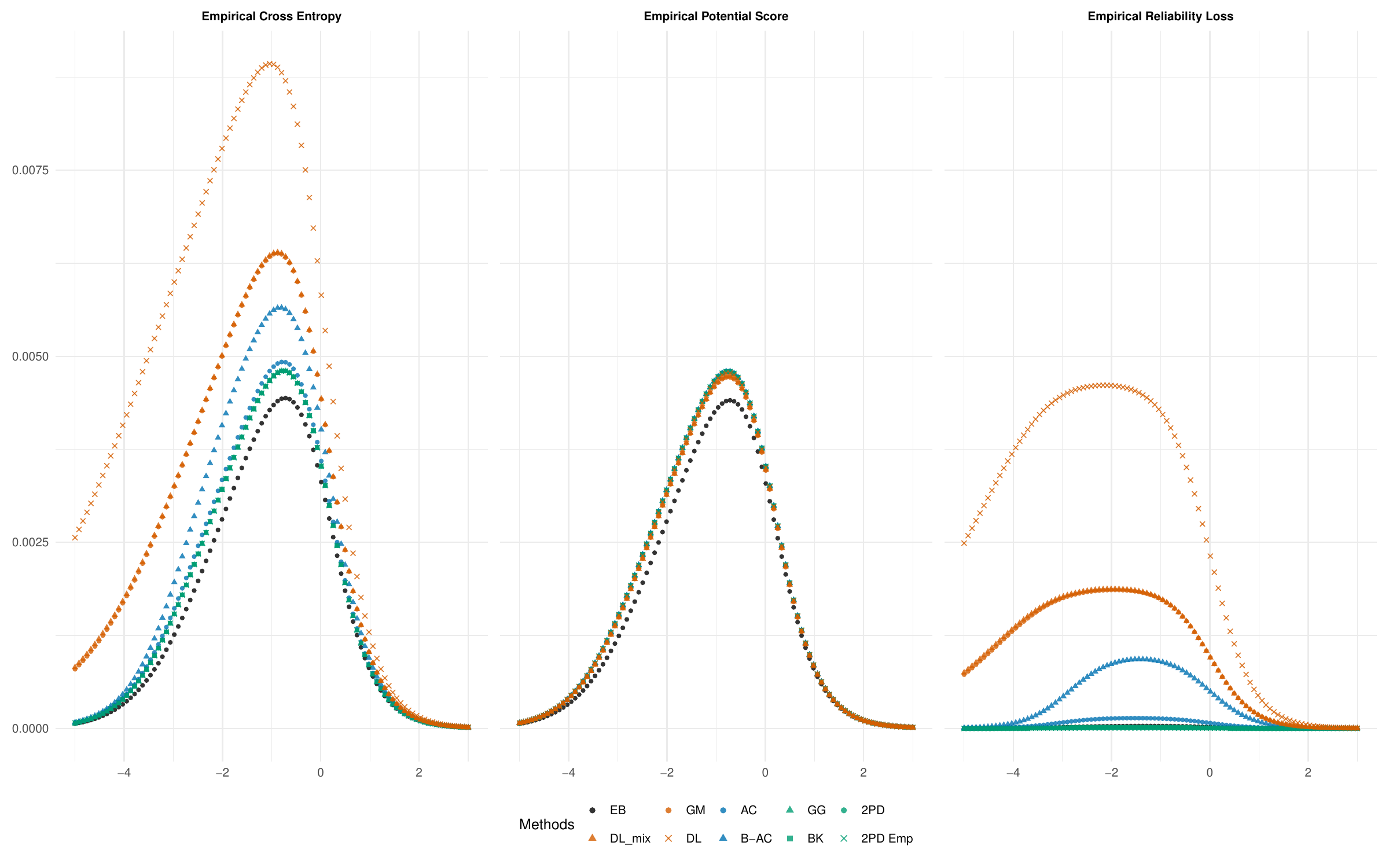}
    \caption{Empirical posterior cross-entropy, potential score, and reliability loss for population 2B. The  methods are colour-coded according to the three groups defined in Table~\ref{tab:met}.}
    \label{fig:combined}
\end{figure*}
\end{center}
The ECE curves exhibit a bell-shaped profile with a slight positive asymmetry (Fig.~\ref{fig:combined}). Their overall shape is largely determined by the potential score, which is remarkably similar across methods, whereas the reliability loss is comparatively flatter but varies substantially between methods. This indicates that the methods have a broadly comparable intrinsic ability to reduce uncertainty about the true hypothesis, while they differ mainly in the extent to which their posterior probabilities are calibrated.

\begin{figure}[htbp]
    \centering
   \includegraphics[width=0.6\linewidth]{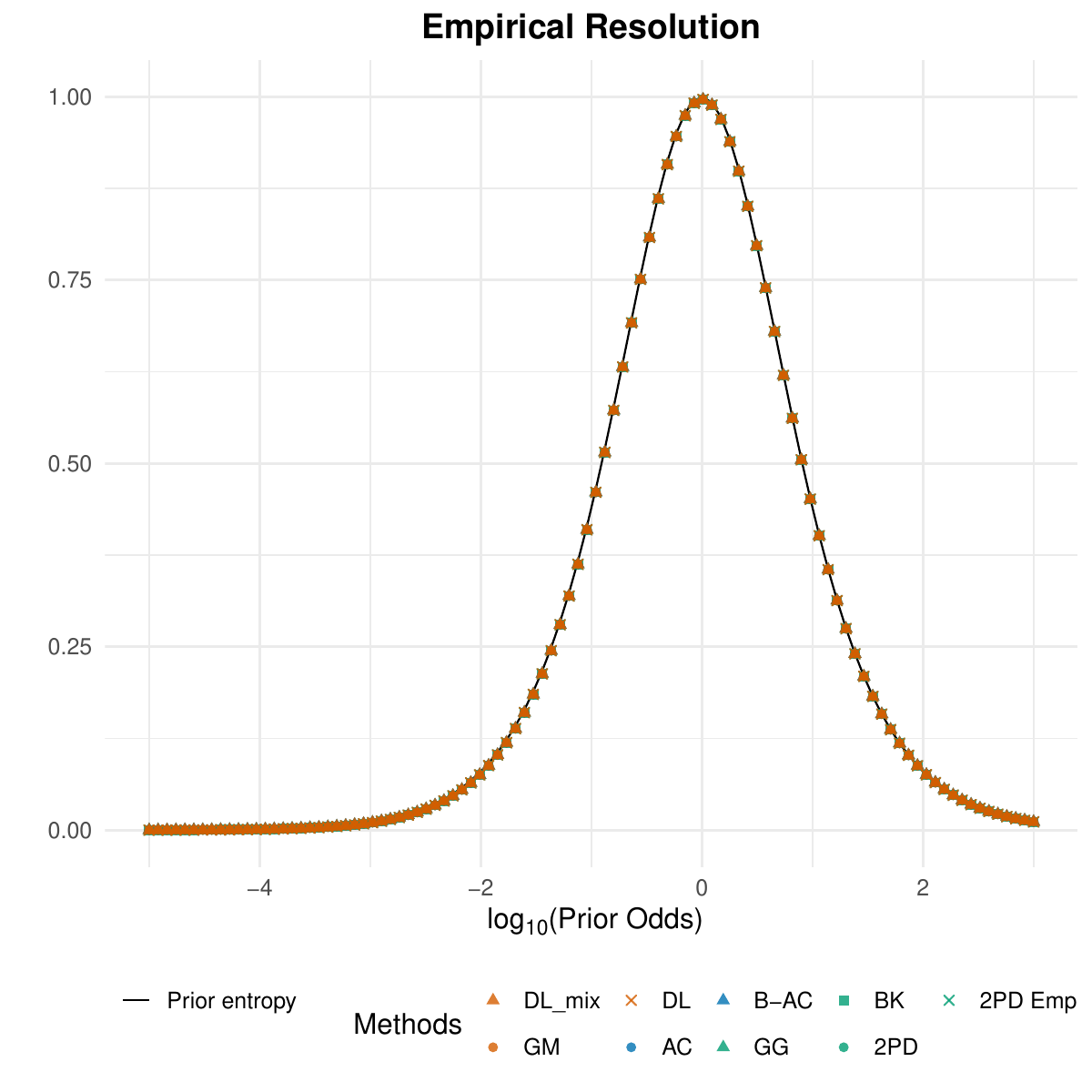}
    \caption{Resolution and prior entropy for population 2B.}
    \label{fig:resolution}
\end{figure}

This interpretation is consistent with Bröcker's decomposition \eqref{eq:brocker}, illustrated in Fig.~\ref{fig:resolution}, where the potential score is expressed as the difference between the prior uncertainty and the resolution. The two quantities are almost indistinguishable for all methods on the scale of the figure, indicating that most of the prior uncertainty is removed by the information contained in the evidence. Their small residual difference corresponds precisely to the potential score shown in Fig.~\ref{fig:combined}. Consequently, if the methods were perfectly calibrated, so that the reliability loss vanished, their performance would approach this common lower bound.

The same conclusion can be viewed more directly through the alternative decomposition proposed in \eqref{eq:ceci3}, where differences in potential score are exactly equal to differences in coarseness. Figure~\ref{fig:delta} therefore reports the pairwise $\Delta$-coarseness between methods. As expected, comparisons involving the EB exhibit substantially larger differences because EB assigns a distinct posterior probability to every haplotype, thereby achieving the minimum attainable coarseness. Excluding the EB, taken as a reference method, pairwise differences among the  methods under scrutiny are remarkably small. Although the methods follow the expected ordering—allele-count methods exhibiting the lowest coarseness, followed by spectrum-based methods and finally allele-frequency methods—the magnitude of these differences remains limited. This indicates that the different strategies for summarizing the evidence lose only a modest amount of discriminatory information.

The principal differences among methods therefore arise from the reliability loss: spectrum-based methods exhibit the smallest reliability loss and consequently the lowest ECE, followed by count-based methods, whereas allele-frequency methods (DL, DL-Mix and GM) show the largest reliability loss. This pattern closely follows the classification proposed in Table~\ref{tab:met}, suggesting that differences in performance are driven primarily by calibration rather than by information loss due to data reduction.

\begin{figure}[htbp]
    \centering
    \includegraphics[width=0.6\linewidth]{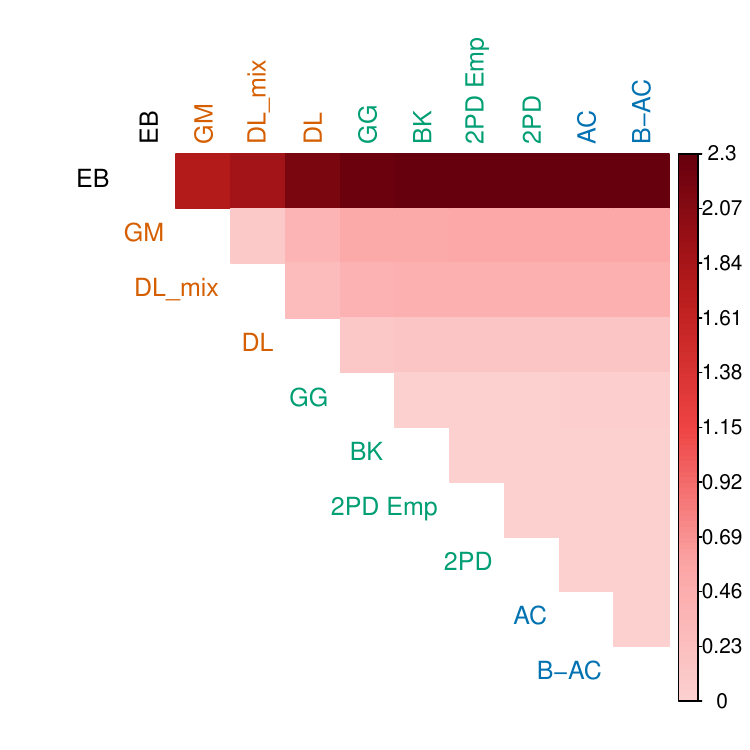}
    \caption{$\Delta\text{Coarseness}_{i-j}$ for each pair of methods $(i,j)$ in population 2B, evaluated under $p(h_p)=p(h_d)=0.5$. Values on the color scale should be multiplied by $10^{-4}$}
    \label{fig:delta}
\end{figure}

To facilitate comparisons across populations, we fix the prior probabilities at $p(h_p)=p(h_d)$=0.5 and examine how the performance measures vary under this common reference scenario.

The ranking of the methods in terms of ECE remains stable across a wide range of prior probabilities (Fig.~\ref{fig:combined}), providing empirical evidence that the comparative evaluation is robust to the choice of the prior on $h_p$.
Therefore, when comparing method performance across populations with different sample sizes and levels of VRS, it is sufficient to consider a single reference prior. Throughout the following analyses, we adopt  $p(h_p)=p(h_d)=0.5$, for which the corresponding ECE is commonly referred to as the \emph{Cllr} (cost of log-likelihood ratio).

Figure~\ref{fig:cllr} reports the \emph{Cllr}, together with its decomposition into potential score and reliability loss, across all populations listed in Table~\ref{tab:population}.  
Populations are ordered along the \( x \)-axis by decreasing size and increasing VRS.  
Moving from left to right, populations become progressively more homogeneous: smaller population sizes correspond to fewer genetically distinct founders, while higher VRS values reflect reproductive dynamics driven by a limited number of highly prolific individuals.  
Both mechanisms contribute to reducing the relative diversity of haplotypes in the population, see Table \ref{tab:Eterogeneity}.

\begin{center}
\begin{table*}[ht!]
\caption{Characteristics of the nine study populations: (I) number of distinct pedigrees, (II) number of distinct haplotypes, and (III) random match probability (RMP). This latter is calculated as in formula \eqref{eq:rmp} in Appendix \ref{sec:rmp}. }
\centering
\resizebox{\textwidth}{!}{%
\begin{tabular}{lrrrrrrrrr}
\toprule
 & 1A & 1B & 1C & 2A & 2B & 2C & 3A & 3B & 3C \\
\midrule
N. of pedigrees   & 3226  & 2129  & 294   & 324  & 207  & 32   & 33  & 25  & 3 \\
N. of haplotypes  & 27584 & 22886 & 11441 & 4040 & 3328 & 1566 & 511 & 449 & 170 \\
RMP  & $5.6\times10^{-5}$ & $6.2\times10^{-5}$ & $8.2\times10^{-5}$
    & $5.3\times10^{-4}$ & $6.1\times10^{-4}$ & $7.9\times10^{-4}$
    & $5.2\times10^{-3}$ & $5.8\times10^{-3}$ & $7.1\times10^{-3}$ \\
\bottomrule
\end{tabular}%
}
\label{tab:Eterogeneity}
\end{table*}
\end{center}

 \begin{center}
\begin{figure*}[ht!]
    \centering
    \includegraphics[width=1\linewidth]{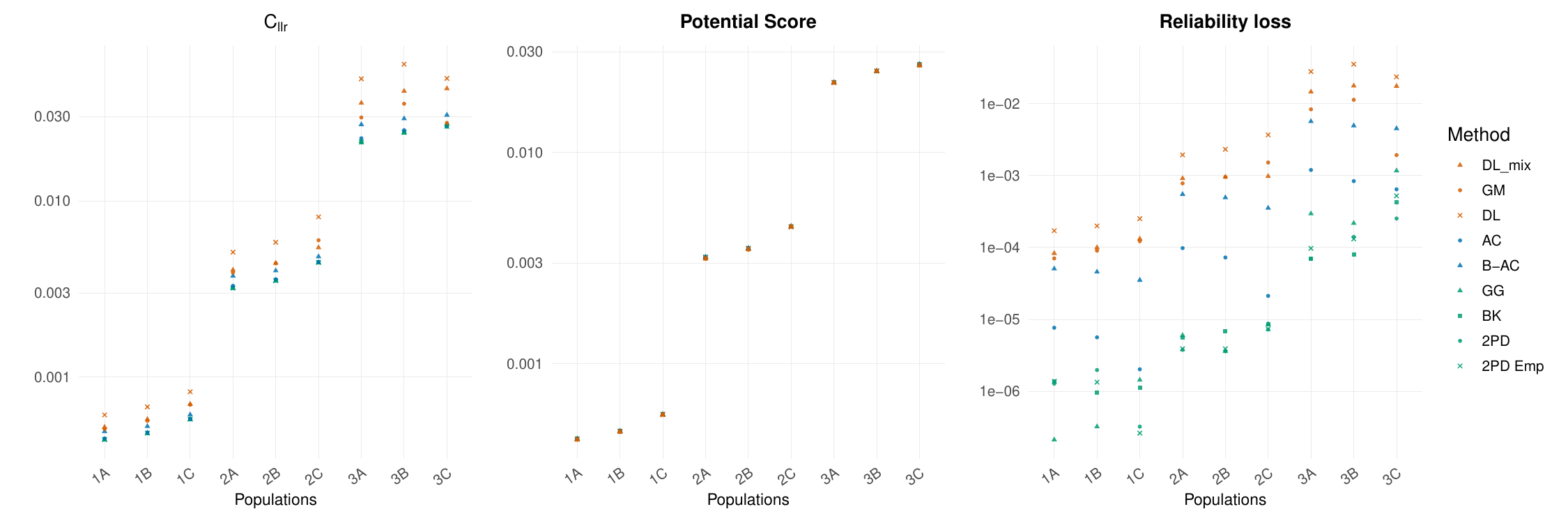}
    \caption{Empirical posterior cross-entropy (\emph{Cllr}), potential score, and reliability loss for all populations in Table~\ref{tab:population}. Methods are colour-coded according to the three groups defined in Table~\ref{tab:met}.}
    \label{fig:cllr}
\end{figure*}
\end{center}

Considering the \emph{Cllr} in Fig. \ref{fig:cllr} as an overall measure of performance, population heterogeneity appears to play a dominant role in determining methods' performance. As population homogeneity increases, the random match probability also increases (see Table~\ref{tab:Eterogeneity}, third row), making coincidental matches under $h_d$	  more frequent. Since these cases contribute most to the \emph{Cllr}, the overall performance deteriorates.
It is apparent that most of the variation in \emph{Cllr} is driven by population size, while VRS has a more limited effect within groups of similar size.  
Again the potential score varies only slightly across methods, indicating that differences in performance are primarily driven by the reliability loss, i.e. by the lack of calibration.

The ranking of the groups remains consistent with that observed for Population 2B in Figure \ref{fig:combined}: the spectrum-based methods achieve the best performance, followed by count-based methods, while allele-frequency-based methods perform the worst, since they achieve higher discrimination at the expense of substantially poorer calibration. 

This ordering may appear counterintuitive, particularly because allele-frequency-based methods attempt to model haplotype probabilities at a finer level of detail. A possible explanation lies in the way these methods construct the haplotype distribution. They first estimate allele frequencies at each locus and then combine them to assign probabilities over the full haplotype space obtained as the Cartesian product of the per-locus allele sets. However, this construction ignores the fact that many --indeed, the vast majority-- of these theoretical haplotypes are absent from the population. Consequently, probability mass is allocated to non-existent haplotypes, systematically reducing the probabilities assigned to the haplotypes that actually occur, including the observed haplotype  
$\tilde{y}$. When the LR is computed according to \eqref{eq:LR}, these underestimated probabilities translate into inflated LR values, particularly detrimental for coincidental matches under \( h_d \). Since such cases contribute most strongly to the \emph{Cllr} and its components, the overall performance of allele-frequency-based methods deteriorates. A possible remedy would be to incorporate population-specific prior knowledge about the haplotype space, assigning negligible probability to implausible haplotypes and thereby improving the calibration of the resulting LR values.\footnote{We note that in the implementation of DL-mix applied to large populations for computational reasons we were forced to limit the number of subpopulations.}
Finally, regarding B-AC we note that  B-AC specifies a $\mathrm{Beta}(1,1)$ prior for $\phi_{\tilde{y}}$ but a more informative prior may be preferable in the rare type match setting, as small values of $\phi_{\tilde{y}}$ are more plausible.

\subsection{Comparison of the methods: only matching cases considered}
\label{sec:match-cases}

Simulated cases can be classified as follows:
\begin{itemize}
   \item \textbf{Non-match cases} ($y_{n+1} \neq y_{n+2}$).  
   In this situation, all methods yield \( \mathrm{LR} = 0 \), correctly voiding the prosecution hypothesis.  
   As a consequence, comparison across methods is not meaningful in this setting, since the same conclusion would be reached by directly observing that \( y_{n+1} \neq y_{n+2} \), without requiring any modelling. These cases arise exclusively when the crime and suspect traces originate from different individuals (i.e.\ under \( h_d \)).  

   \item \textbf{Match cases}, for which \( y_{n+1} = y_{n+2} \).  These cases  arise both when the traces originate from the same individual (under \( h_p \)) or from coincidental matches (under \( h_d \)). In this cases, all methods produce LR$>0$, and so it is interesting to study  how well  they are able to distinguish between the two underlying hypotheses according to the value of the LR.  
   This ability is especially relevant for coincidental matches, which lead to substantially misleading posterior assessments. 
\end{itemize}

\begin{figure}
    \centering
    \includegraphics[width=0.95\linewidth]{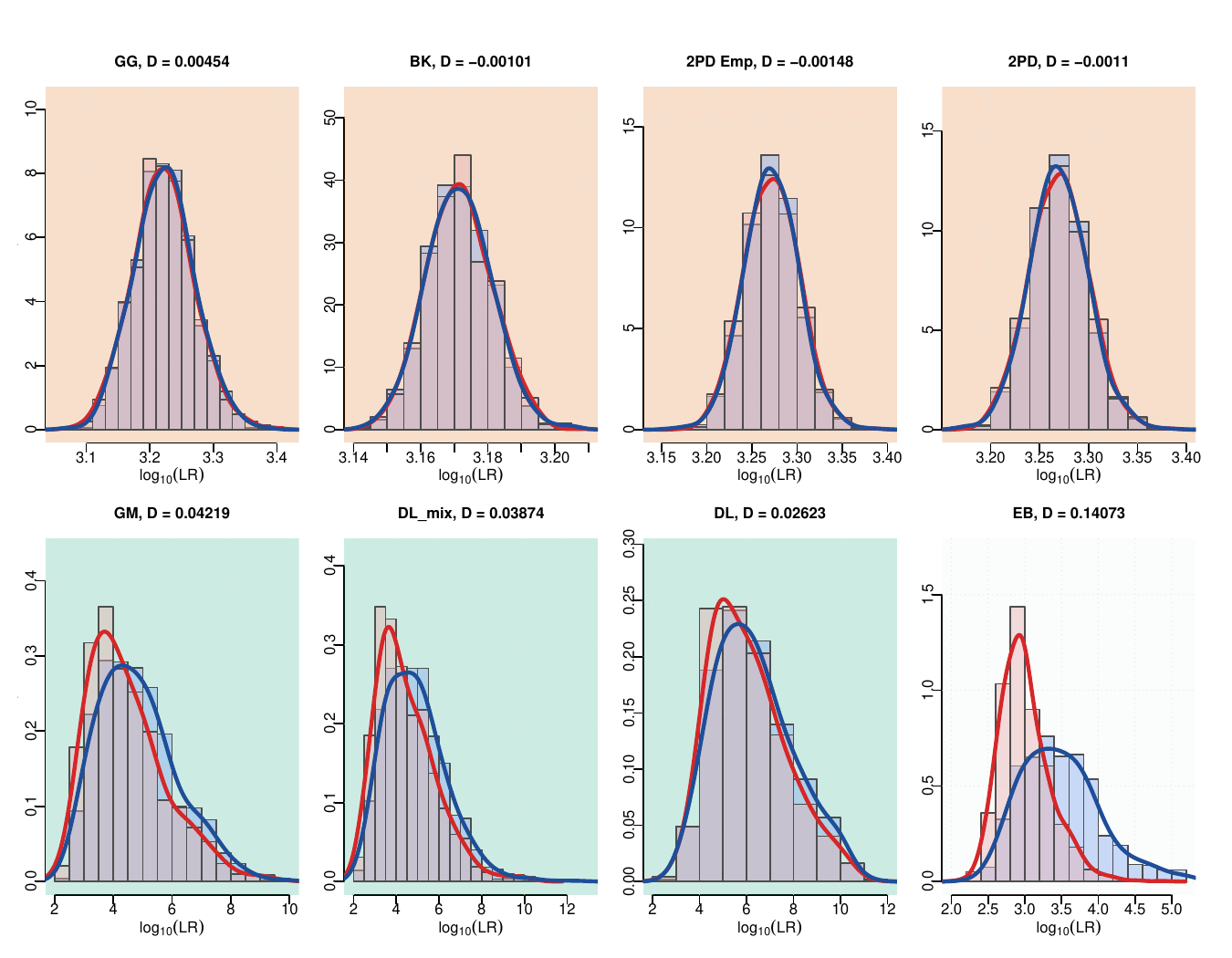}
    \caption{LR distribution for simulated $h_p$ (red) and $h_d$ (blue) cases, with specification of the value of the directional CDF $D$, net to the name of each method, for population 2B. }
    \label{fig:placeholder}
\end{figure}

In Fig.~\ref{fig:placeholder}, each method is represented by two histograms showing the distribution of LR values for match cases generated under \( h_p \) and \( h_d \), respectively.  
For comparison, we also include the EB, while  AC and B-AC are not represented since they  produce exactly the same  LR value for all match cases and under both hypotheses, thus exhibiting no discriminatory ability. Even though a perfect LR system would assign $0<\textrm{LR}<1$ to $h_d$ cases and $\textrm{LR}>1$ to $h_p$ cases, we would be satisfied if the two sets of LR values did not overlap. 
To quantify the degree of separation between the two distributions (under $h_p$ and $h_d$), we introduce the \emph{directional CDF distance}, defined as
\begin{equation}
\label{eq:dcdf}
    \hat{D} = \frac{1}{x_{\max} - x_{\min}}
           \sum_{k=1}^K \bigl[\hat{F}_{h_d}(x_k) - \hat{F}_{h_p}(x_k)\bigr]
           \,\Delta x,
\end{equation}
where \( \hat{F}_{h}(\cdot) \), \( h \in \{h_p, h_d\} \), denotes the empirical CDF.

This measure takes values in \( [-1,1] \):  
a value \( \hat{D} = 0 \) indicates complete overlap between the two distributions, while \( \hat{D} = 1 \) corresponds to perfect separation, where LR values under \( h_p \) are always larger than those under \( h_d \).  
Negative values indicate tendency to misclassification. 

 Across all methods and for both hypotheses, the minimum LR value (on the natural scale) exceeds 100.  
    This indicates that, in  case of coincidental matches, all methods yield strongly misleading results in favor of the prosecution hypothesis.

  While spectrum-based methods show limited variability in LR values and substantial overlap between the distributions under \( h_p \) and \( h_d \), suggesting weak discriminatory power, 
 allele-frequency-based methods and the EB benchmark exhibit a more pronounced separation between the two distributions.   This  can be explained by the data-generating mechanism.  
    The haplotype \( y_{n+1} \) is sampled among those not present in the reference database and is therefore expected to have low population probability.  
    Under \( h_p \), LR values reflect the inverse of its probability.  
    Under \( h_d \), a match occurs only if a second independent draw from the population equals \( y_{n+1} \), which is more likely for haplotypes with higher population probabilities.  
    This selection mechanism results in smaller LR values for coincidental matches than for matches under \( h_p \), thereby inducing separation between the two distributions. The more the data are reduced the weaker   the separation is.

In conclusion, from the perspective of hypotheses discrimination in case of a match, allele-frequency-based methods appear to be the most promising.  
This suggests that, following appropriate methodological refinements (as discussed in Section~\ref{sec:all-cases}), further development of this class of methods may be worthwhile.

The previous results describe the overall behaviour of the ECE. To understand which cases drive these differences, we now examine separately the contributions to ECE of simulations under $h_p$
 and under $h_d$.

Cases under $h_p$ always 
correspond to genuine matches and therefore yield LRs largely greater than 1, with 
limited associated penalties. By contrast, under $h_d$, coincidental 
matches may occur — in which an individual unrelated to the suspect 
happens to share the same haplotype — resulting in positive LRs and 
potentially large penalties.
 As a consequence 
the $h_d$ component is the dominant one.

Figure~\ref{fig:perc_hd} shows the percentage of the ECE attributable to the (very few) coincidental matches under $h_d$. As it appears the allele-frequency methods suffer enormously of the  $h_d$-cases contribute and this effect is due to their lack of reliability which massively affects  ECE.  Spectrum based methods and allele-count methods follow.

\begin{figure}[htbp]
    \centering
    \includegraphics[width=0.7\linewidth]{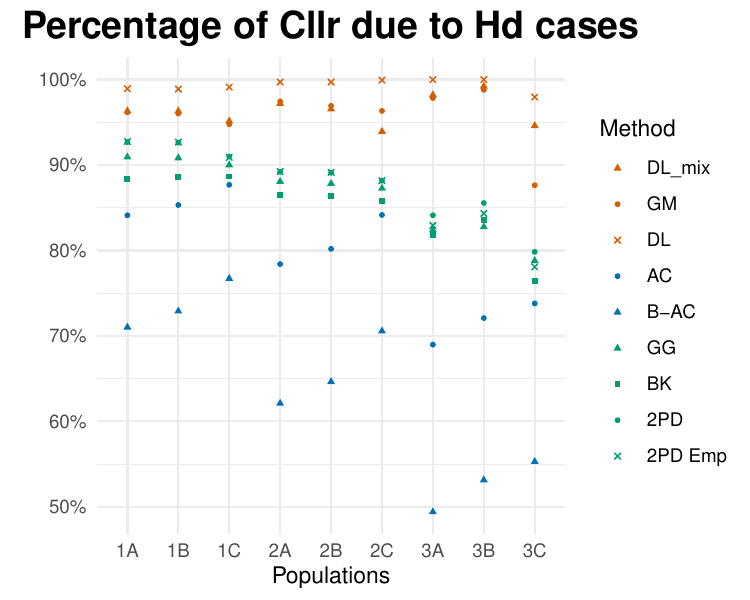}
    \caption{Percentage of the empirical posterior cross-entropy (ECE) 
    attributable to cases simulated under $h_d$, across the nine simulated 
    populations.}
    \label{fig:perc_hd}
\end{figure}

\section{Discussion and conclusion}
\label{sec:conclusions}
In this work, we analysed and compared several methods currently adopted in forensic practice across different countries \citep{andersen:2021} for evaluating the likelihood ratio (LR) in rare type cases \citep{cereda:2017}.  

While previous studies \citep{andersen:2021,brenner:2014} have reviewed methods for assessing the evidential value of Y-STR profile matches, to the best of our knowledge this is the first work specifically aimed at defining a general framework for comparing methods designed for the rare type match problem.

A first contribution of the paper is a structured description of the methods, emphasizing the different types of summary statistics upon which they rely to compute the LR.  
This leads naturally to a classification into three groups: 
the count-based methods,  
the spectrum-based methods, and the allele-frequency-based methods, detailed in Table \ref{tab:met}.

The second contribution is the definition of a framework for evaluating the considered methods based on the posterior cross-entropy, interpreted as the expected cost of soft decisions obtained by combining a prior distribution on the hypotheses with the LR produced by the method,  under the logarithmic scoring rule.
To analyse the posterior cross-entropy, we propose a reformulation of the problem by shifting the focus from the original haplotype space to the space induced by the statistics employed by each method.
This change of perspective highlights how these statistics affect the expected cost, offering a unified framework to evaluate the methods through both two existing decompositions and a newly proposed one.  
The first decomposition allows us to reinterpret the posterior cross-entropy in terms of  reliability loss and potential score. The reliability loss measures the discrepancy between the method's posterior probabilities and their calibrated counterparts, while the potential score quantifies the maximum  separation between the hypothesis achievable by these calibrated probabilities. This decomposition highlights that the reliability loss may, in principle, be reduced by recalibration, whereas the potential score represents an intrinsic characteristic of the method that cannot be improved.
It should be noted that this decomposition, arising from the proposed change in perspective, resembles the one proposed by \citet{degroot:1982}. However, in their work, the posterior probabilities were assumed to belong to an ad hoc finite set, whereas in our framework the corresponding set is naturally induced by the method through the space of its summary statistics. Consequently, no arbitrary discretization is required. 
We also re-frame a second decomposition originally proposed by \citet{brocker:2009}. It shows how well-calibrated posterior probabilities reduce the uncertainty of the prior distribution through the information conveyed by the summary statistics (see Fig.~\ref{fig:resolution}). 
We further propose a novel decomposition introducing the notion of coarseness. This quantity explicitly measures the information loss induced by the granularity of the summary statistics employed by the methods: the finer the granularity, the smaller the coarseness.

A third contribution of the paper concerns the empirical estimation of posterior cross-entropy, and of the terms of its decompositions, within a simulation framework.  
The main challenge lies in obtaining reliable estimates of the calibrated probabilities associated with each method.  
The quality of these estimates depends both on the size of the reference database and on the cardinality of the space of the statistics,  \( \mathcal{X}_m \).  
For some methods, the latter becomes so large that empirical estimation is unstable and the law of large numbers operates poorly.  
To address this problem, and following \citet{brummer:2010}, we reduced the effective cardinality of \( \mathcal{X}_m \) by means of the Pool Adjacent Violators (PAV) algorithm.  
Among monotonic calibration procedures, PAV provides a conservative and statistically efficient estimate of the Potential score.

Finally, as a last contribution, the methods were compared through extensive simulations.  
Following the framework proposed by \citet{andersen:2021}, we simulated populations with three different sizes under three different levels of variance in reproductive success (VRS),  resulting in a total of nine populations.  
 A first point that emerges is that, in the presence of a match, all methods produce LR values substantially larger than one, irrespective of the true hypothesis.  
Consequently, even coincidental matches under \( h_d \) generated strong support in favour of the prosecution hypothesis.  
Conversely, in non-match cases, all methods produced an LR of 0, correctly supporting  \( h_d \).  This behavior implies that both resolution and coarseness, the only components of the potential score that depend on how the methods' posterior probabilities partition the sample space, are relatively similar across methods.
The main differences, instead, arise from the reliability loss, particularly in the presence of coincidental matches. In these cases, allele-frequency-based methods exhibit substantially greater miscalibration, resulting in larger ECE. This may be explained by the fact that these methods spread probability mass over many hypothetical haplotypes, many of which may not actually exist in the population, thereby systematically underestimating the probabilities of existing haplotypes.
This in turn produces overestimated LR values in  case of matches.
At the same time, allele-frequency-based methods were the only methods exhibiting a non-negligible ability to discriminate between genuine and coincidental matches.  
This property originates from their attempt to estimate haplotype probabilities directly.  
Under coincidental matches, a double selection mechanism operates, making haplotypes with relatively larger population probabilities more likely to produce matches under \( h_d \).  
This induces partial separation between the LR distributions obtained under the two hypotheses.

The simulation framework also allowed us to evaluate the maximum discrimination theoretically achievable when the true haplotype distribution is known.  
The results suggest that substantial room for improvement remains.  
Further research devoted to the development of methods capable of accurately estimating haplotype distributions therefore appears warranted.

Finally, by ordering populations according to their degree of haplotype heterogeneity — from large populations with low VRS (high heterogeneity) to small populations with high VRS (high homogeneity) — we observed a systematic deterioration in methods' performance as measured by ECE.  
More homogeneous populations exhibit a higher probability of coincidental matches, increasing the frequency of misleading evidence in favor of \( h_p \).  
This finding further supports the need for methods capable of accurately modelling haplotype probabilities within the relevant reference population.

These conclusions are consistent with the considerations of \citet{andersen:2021}  who argue that Y-STR haplotypes should not be evaluated with respect to excessively broad populations unrelated to the likely genealogy of the trace donor.  
Rather, the relevant population should be narrowed also according to the investigative context.  
If supported by methods capable of accurately estimating haplotype probabilities, this approach would likely produce substantially smaller and better calibrated LR values than those obtained by considering very large and heterogeneous populations.

\subsection*{Financial Disclosure}

None reported.

\subsection*{Conflicts of Interest}

The authors declare no conflicts of interest.
\subsection*{Acknowledgment}
This work has had a long gestation, and we are grateful to the many people who supported its development over the years. In particular, we thank Mikkel Meyer Andersen for initiating the project during a discussion at the 2016 research program on Probability and Statistics in Forensic Science, hosted by the Isaac Newton Institute for Mathematical Sciences at the University of Cambridge, and for developing the first prototype code. We also thank Richard Gill and Marjan Sjerps for recognizing the potential of the idea, for their encouragement and support, and for motivating the first author to prepare an initial version of the manuscript. We are grateful to Tóra Oluffa Stenberg Olsen for carefully reviewing the section on the discrete Laplace distribution and for providing valuable suggestions that helped us improve the final version of the manuscript. We also thank Cosimo Grazzini for carefully reviewing the appendix. AI tools (ChatGPT and Gemini) assisted with manuscript editing and translation.

\bibliographystyle{apalike}
\bibliography{bibliography}

\subsection*{Supporting Information}

Additional supporting information can be found online in the Supporting Information
section.

\appendix

\section{Technical Details}
\vspace*{12pt}
\begin{proposition}
\label{app:A1}
Let $\mathcal{X}_m$ denote the finite set of values assumed by the posterior distribution $p_m(h\mid y)$ when $y$ varies in $ \mathcal{Y}$, and for each $x\in \mathcal{X}_m$, let $\mathcal{Y}^x_m: =\{y\in \mathcal{Y} \mid p_m(h_p\mid y)=x\}.$ 
It holds 
\begin{enumerate}
    \item $p(h\mid \mathcal{Y}_m^x)=\frac{\sum_{y\in \mathcal{Y}_m^x}p(h\mid y)p(y)}{\sum_{y\in \mathcal{Y}_m^x}p(y)}$ if $p(y)>0$ for each $y\in \mathcal{Y}_m^x$;
    \item $p_m(h_p \mid \mathcal{Y}^x_m)=x$ if $p_m(y)>0$ for each $y\in \mathcal{Y}_m^x$.

\end{enumerate}
\end{proposition}
\begin{proof}
Using the third axiom of probability it holds
$p(\mathcal{Y}^x_m)= \sum \limits_{y\in \mathcal{Y}^x_m } p(y) $
and $p(\mathcal{Y}^x_m\mid h)= \sum \limits_{y\in \mathcal{Y}^x_m } p(y\mid h) \text{ for } h\in\{h_p,h_d\}.$
Thus, for each $h\in\{h_p, h_d\}$ we have 
$${p(h, \mathcal{Y}_m^x)}=p(h) \sum \limits_{y\in \mathcal{Y}^x_m } p(y\mid  h)=\sum \limits_{y\in \mathcal{Y}^x_m} p(h,y)= \sum \limits_{y\in \mathcal{Y}^x_m} p(h\mid y) p(y).$$
\begin{enumerate}
    \item $p(h\mid \mathcal{Y}_m^x)=\frac{p(h, \mathcal{Y}_m^x)}{p(\mathcal{Y}_m^x)}=\frac{\sum_{y\in \mathcal{Y}_m^x}p(h\mid y)p(y)}{\sum_{y\in \mathcal{Y}_m^x}p(y)}$;
    \item $p_m(h_p\mid \mathcal{Y}_m^x)=
    \frac{\sum_{y\in \mathcal{Y}_m^x}p_m(h_p\mid y)p_m(y)}{\sum_{y\in \mathcal{Y}_m^x}p_m(y)}= \frac{\sum_{y\in \mathcal{Y}_m^x} x \cdot p_m(y)}{\sum_{y\in \mathcal{Y}_m^x}p_m(y)}=x \cdot\frac{\sum_{y\in \mathcal{Y}_m^x}  p_m(y)}{\sum_{y\in \mathcal{Y}_m^x}p_m(y)}=x,$  where we exploit, by definition, $p_m(h_p\,|\,y)=x$ for each $y \in \mathcal{Y}_m^x$.
    
\end{enumerate}
\end{proof}

\begin{proposition}
\label{app:A2}
   For every method $m$ with probability distribution $p_m$, the posterior cross entropy can be written as
    
\begin{align*} &\mathcal{P}CE(p(h\mid y); p_m(h\mid y))= 
\mathbb{E}_{X_m}(D(p(h\mid x); p_m(h\mid x)))+\mathbb{E}_{X_m} ( H(p(h\mid x))).
\end{align*}
 if $p(y)>0 \, \forall y \in \mathcal{Y}$ and $p(h)>0 \, \forall \, h \in \{h_p,h_d\}$.
\end{proposition}
\begin{proof}
\begin{align}
    \mathcal{P}CE(p(h\mid y); p_m(h\mid y))
&=- \sum_{h\in\{h_p,h_d\}} \sum_{y \in \mathcal{Y}} p(y)p(h \mid y)  \log p_m(h\mid y) \nonumber \\
&=- \sum_{h\in\{h_p,h_d\}} \sum_{x\in \mathcal{X}_m} \sum_{y \in \mathcal{Y}_m^x} p(y\mid h)p(h)  \log p_m(h\mid y) \nonumber \\
&=- \sum_{h\in\{h_p,h_d\}} \sum_{x\in \mathcal{X}_m} p(\mathcal{Y}_m^x\mid h)p(h)   \log p_m(h\mid \mathcal{Y}_m^x) \label{eq:pcalm} \\
&=- \sum_{h\in\{h_p,h_d\}} \sum_{x\in \mathcal{X}_m} p(h \mid \mathcal{Y}_m^x)  p(\mathcal{Y}_m^x)  \log p_m(h\mid \mathcal{Y}_m^x)\nonumber \\
&=- \sum_{x\in \mathcal{X}_m}  p(\mathcal{Y}_m^x)   \sum_{h\in\{h_p,h_d\}}p(h \mid \mathcal{Y}_m^x) \log p_m(h\mid \mathcal{Y}_m^x) \label{eq:rephrPCE},
\end{align}
where \eqref{eq:pcalm} comes from the fact that $p_m(h_p\mid y)=p_m(h_p \mid \mathcal{Y}^x_m )=x$ for each $y\in \mathcal{Y}^x_m$
as proven in Proposition \ref{app:A1}.
 By adding and subtracting $\log(p(h\mid \mathcal{Y}_m^x))$ from \eqref{eq:rephrPCE}, we obtain
\begin{align}
 \mathcal{P}CE(p(h\mid y); p_m(h\mid y)) 
=&- \sum_{x\in \mathcal{X}_m}  p(\mathcal{Y}_m^x)   \sum_{h\in\{h_p,h_d\}}p(h \mid \mathcal{Y}_m^x) \Bigl(\log p_m(h\mid \mathcal{Y}_m^x) \pm \log p(h\mid \mathcal{Y}_m^x) \Bigl) \nonumber  \\
=& \sum_{x\in \mathcal{X}_m}  p(\mathcal{Y}_m^x)   \sum_{h\in\{h_p,h_d\}}p(h \mid \mathcal{Y}_m^x) \log\Bigl(\frac{p(h\mid \mathcal{Y}_m^x)}{p_m(h\mid \mathcal{Y}_m^x)} \Bigl) - \sum_{x\in \mathcal{X}_m}  p(\mathcal{Y}_m^x)   \sum_{h\in\{h_p,h_d\}} p(h \mid \mathcal{Y}_m^x)\log p(h\mid \mathcal{Y}_m^x) \nonumber\\
=& \mathbb{E}_{X_m}(D(p(h\mid  \mathcal{Y}_m^x); p_m(h\mid  \mathcal{Y}_m^x)))+  \mathbb{E}_{X_m} ( H(p(h\mid  \mathcal{Y}_m^x))). \nonumber
\end{align}
\end{proof}

\begin{proposition}
\label{app:A3}
    For every $p_m$ the posterior cross entropy can be decomposed as follows:
\begin{align}
&\mathcal{P}CE(p(h\mid y); p_m(h\mid y))= \nonumber\\
&=\mathbb{E}_{X_m}(D(p(h\mid x); p_m(h\mid  x)))+ H(p(h))- \mathbb{E}_{X_m}(D(p (h\mid x);p(h))).\nonumber \label{eq:Abrokerdec} 
\end{align}
 if $p(y)>0 \, \forall y \in \mathcal{Y}$ and $p(h)>0 \, \forall \, h \in \{h_p,h_d\}$.

\end{proposition}

\begin{proof}
Following from Proposition \ref{app:A2}, it is sufficient  to prove that  $$H(p(h))-\mathbb{E}_{X_m}(D(p(h\mid \mathcal{Y}_m^x);p(h)))= \mathbb{E}_{X_m} ( H(p(h\mid  \mathcal{Y}_m^x)))$$

\begin{align}
    &H(p(h))-\mathbb{E}_{X_m}(D(p(h\mid x);p(h))) = \nonumber\\
    &= -\sum_{h\in\{h_p,h_d\}} p(h) \log p(h)- \sum_{x\in \mathcal{X}_m} p(\mathcal{Y}^x_m)\sum_{h\in\{h_p,h_d\}}p(h \mid \mathcal{Y}^x_m)\log \left(\frac{p(h\mid \mathcal{Y}^x_m)}{ p(h)} \right) \nonumber \\
    &= -\sum_{h\in\{h_p,h_d\}} p(h) \log p(h) - \sum_{x\in \mathcal{X}_m} p(\mathcal{Y}^x_m)\sum_{h\in\{h_p,h_d\}}p(h \mid \mathcal{Y}^x_m) \log p(h\mid \mathcal{Y}^x_m) +\sum_{x\in \mathcal{X}_m} \sum_{h\in\{h_p,h_d\}} p(h \mid \mathcal{Y}^x_m) p(\mathcal{Y}^x_m) \log p(h)  \nonumber\\
    &= -\sum_{h\in\{h_p,h_d\}} p(h) \log p(h)- \sum_{x\in \mathcal{X}_m} p(\mathcal{Y}^x_m)\sum_{h\in\{h_p,h_d\}}p(h \mid \mathcal{Y}^x_m)\log p(h\mid \mathcal{Y}^x_m) + \sum_{h\in\{h_p,h_d\}} \log p(h) \sum_{x\in \mathcal{X}_m}  p(h ,\mathcal{Y}^x_m)   \nonumber\\
    &= -\sum_{h\in\{h_p,h_d\}} p(h) \log p(h)- \sum_{x\in \mathcal{X}_m} p(\mathcal{Y}^x_m)\sum_{h\in\{h_p,h_d\}}p(h \mid \mathcal{Y}^x_m)\log p(h\mid \mathcal{Y}^x_m) + \sum_{h\in\{h_p,h_d\}} \log p(h) p(h) \nonumber\\
    &= - \sum_{x\in \mathcal{X}_m} p(\mathcal{Y}^x_m)\sum_{h\in\{h_p,h_d\}}p(h \mid \mathcal{Y}^x_m)\log p(h\mid \mathcal{Y}^x_m) \nonumber \\ 
    &= \mathbb{E}_{X_m} ( H(p(h\mid  \mathcal{Y}_m^x))). \nonumber
\end{align} 
 
\end{proof}

\begin{proposition}
\label{App:A4}
The posterior cross entropy, between the reference posterior, $p(h\mid y)$, and the posterior distribution provided by the method $m$, $p_m(h \mid y)>0$, can be alternatively decomposed into the following sum:

\begin{align}
  \mathcal{P}CE(p(h\mid y); p_m(h\mid y))  &=\mathbb{E}_Y(H(p(h\mid y))) +   \mathbb{E}_Y(D(p(h\mid y); p(h\mid x)))  + \mathbb{E}_{X_m}(D(p(h\mid x); p_m(h\mid  x))).  \nonumber
\end{align}
\end{proposition}

\begin{proof}
Following from Proposition \ref{app:A2}, it is sufficient  to prove that   $$\mathbb{E}_Y(H(p(h\mid y))) +   \mathbb{E}_Y(D(p(h\mid y); p(h\mid x)))= \mathbb{E}_{X_m} ( H(p(h\mid x))).$$
\begin{align}
\mathbb{E}_Y(H(p(h\mid y)))  +   \mathbb{E}_Y(D(p(h\mid y); p(h\mid x))) &=-\sum_y p(y)\sum_h p(h\mid y)\log  p(h\mid y)  +\sum_y p(y)\sum_h p(h\mid y)\log \frac{p(h\mid y)}{p(h\mid \mathcal{Y}_m^x)} \nonumber \\
&=-\sum_y p(y)\sum_h p(h\mid y)\log  p(h\mid \mathcal{Y}_m^x)  \nonumber \\
&=-\sum_h p(h)\sum_{X_m} \sum_{y \in \mathcal{Y}_m^x} p(y\mid h)\log  p(h\mid \mathcal{Y}_m^x) \nonumber \\
&=-\sum_h p(h)\sum_{x\in \mathcal{X}_m}    p(\mathcal{Y}_m^x\mid h)\log  p(h\mid \mathcal{Y}_m^x) \nonumber \\
&=- \sum_x p(\mathcal{Y}_m^x) \sum_h p(h\mid\mathcal{Y}_m^x )    \log  p(h\mid \mathcal{Y}_m^x) \nonumber\\
&= \mathbb{E}_{X_m} ( H(p(h\mid  \mathcal{Y}_m^x))). \nonumber
\end{align}

\end{proof}

 \section{How to re-calibrate posterior distribution}\label{sec:cal}
 \vspace*{12pt}
 
Calibration techniques are methods that aim at estimating the so-called \emph{calibration function},  transforming a posterior $p_m$ into an estimate $p_m^{cal}$ of its well-calibrated version $p(h\mid \mathcal{Y}_m^x)$. 
Clearly, the estimation would lead to a more reliable posterior, but never attains the perfectly well-calibrated posterior. 

If the size of each $\mathcal{Y}_m^x$ are large enough, it is always possible to proceed to what we will refer to as the \emph{empirical calibration}, which is a Monte Carlo estimate of $p(h| \mathcal{Y}^x_m)$ for every $x \in \mathcal{X}_m$:

\begin{eqnarray}
\label{eq:empcal}
\widehat{p}_m^{emp}(h_p|\mathcal{Y}^{x}_m)=\frac{\sum \limits_{i=1}^N \mathbb{1} \{h^{(i)}=h_p,y^{(i)} \in \mathcal{Y}^{x}_m \}}{\sum \limits_{i=1}^N \mathbb{1}\{y^{(i)} \in \mathcal{Y}^{x}_m \}},\quad \forall x \in \mathcal{X}_m.
\end{eqnarray}

However, if the number of observations in each set $\mathcal{Y}_m^x$ is not large enough to guarantee a good estimate of \eqref{eq:empcal}, the solution consists
of aggregating cases with similar values of $x$  and evaluating a single calibrated probability for all of them.

The issue is to find a proper aggregating rule. In the empirical binning methods \citep{zadrozny:2001, naeini:2015}, the bins are chosen with some ad hoc, predefined rule, but the results strongly depend on this rule.

Another solution is to use isotonic regression as a calibration approach \citep{barlow:1972, niculescu-mizil:2005, naeini:2016}. These methods learn the bin boundaries from the data to provide calibrated probabilities which comply with the ranking of the corresponding posterior value.
 
 The Pool Adjacent Violator (PAV) \citep{ayer:1955} is an algorithm to perform isotonic regression that optimizes the posterior cross entropy subject to the monotonicity constraint.

\subsection{PAV algorithm}
 
 Suppose, without loss of generality that the simulated posterior probability values $p_i:= p(h_p\mid y^{(i)})$ can be ordered in a sequence 
 $$ 0< p_1\leq p_2\leq ... \leq p_n < 1.$$
 
The algorithm generates a sequence of (re)calibrated  and ordered posteriors
$$0<\hat{p_1}\leq \hat{p_2}\leq ... \leq \hat{p_n}<1$$ starting with the sequence of empirical posteriors (as in Equation \eqref{eq:empcal}) $\tilde{p_1}, \tilde{p_2}, ..., \tilde{p_n}$, and replacing every $\tilde{p_i}$ that violates monotonicity ($\tilde{p_i}> \tilde{p}_{i+1})$ and all the subsequent equal values with the empirical frequencies of the corresponding $h_i=\mathbb{1}(h^{(i)}=h_p)$.

It follows that after calibration of the PAV $\mathcal{X}_{m,cal}\neq \mathcal{X}_{m}$ and in particular $|\mathcal{X}_{m,cal}|\leq |\mathcal{X}_{m}|$. 

Notice that the algorithm was developed to solve the isotone regression task: given a training set $(p_i, h_i)$ finding a map $m:p_i\rightarrow \hat{p_i}:=m(p_i)$, with corresponding non-negative weights $w_i$, that minimizes \begin{equation}\label{p}
    \sum_i^n w_i(h_i-\hat{p_i})^2,
\end{equation} and preserve the order of the $p_i$. 
As evident from \eqref{p}, such a map will optimize the empirical average Brier score but it can be proved that it optimizes all binary proper scoring rules  \citep{brummer:2013}.

We use a  slightly modified version of the standard PAV algorithm  that handles ties as described by \cite{deleeuw:2009}, ensuring that equal values of $x_m$	
  receive equal calibrated values $x_m^{cal}$.

\section{Probability of a random match in the rare type match case}\label{sec:rmp}
\vspace*{12pt}
Let $i\in \{1,..., K\}$ index the different haplotypes in the simulated population and $n$ the size of the sampled dataset $ \mathcal{Y}_n= \{y_1,...,y_n\}$. The following holds only assuming sample with replacement as an approximation of the real one that is without replacement.

The following result holds under i.i.d. sampling with replacement, used here as an approximation to the actual sampling scheme without replacement

\begin{align}
\label{eq:rmp}
\Pr(Y_{n+2}=Y_{n+1}\mid Y_{n+1}\notin \mathcal{Y}_n)&=\frac{\Pr(Y_{n+2}=Y_{n+1}, Y_{n+1}\notin \mathcal{Y}_n)}{\Pr( Y_{n+1}\notin \mathcal{Y}_n)} \notag \\&=
\frac{\sum_{i=1}^K \Pr(Y_{n+2}=Y_{n+1}, Y_{n+1}\notin \mathcal{Y}_n\mid Y_{n+1}=i)\Pr(Y_{n+1}=i)}{\sum_{i=1}^K  \Pr( Y_{n+1}\notin \mathcal{Y}_n\mid Y_{n+1}=i)\Pr(Y_{n+1}=i)}\notag\\&=
\frac{\sum_{i=1}^K \Pr(Y_{n+2}=i, i \notin \mathcal{Y}_n\mid Y_{n+1}=i)\Pr(Y_{n+1}=i)}{\sum_{i=1}^K  \Pr( i\notin \mathcal{Y}_n\mid Y_{n+1}=i)\Pr(Y_{n+1}=i)}
\notag\\&=
\frac{\sum_{i=1}^K \Pr(Y_{n+2}=i ) \Pr(i \notin \mathcal{Y}_n)\Pr(Y_{n+1}=i)}{\sum_{i=1}^K  \Pr( i\notin \mathcal{Y}_n)\Pr(Y_{n+1}=i)}\notag\\&=
\frac{\sum_{i=1}^K f_i \times (1-f_i)^n\times f_i}{\sum_{i=1}^K (1-f_i)^n\times f_i}\\ \nonumber &= 
\frac{\sum_{i=1}^K   (1-f_i)^n f_i^2}{\sum_{i=1}^K (1-f_i)^n f_i} 
\end{align}

 where $f_i$ 	denote the population frequency of haplotype $i$, with $\sum f_i=1$, and in the second to last equality, we exploited the independence property due to the sampling with replacement.

\section{Br\"ocker's bias}
\label{App:H}
\vspace*{12pt}
 As reported in \cite{brocker:2012}, the accuracy of these empirical quantities as estimates of their theoretical counterparts depends both on the sample size $N$ and the cardinality of the output space $|\mathcal{X}_m|$. Notably, they are not unbiased estimators. Specifically, the empirical reliability and resolution tend to overestimate their true values, while the empirical potential score tends to underestimate it:

 $$\mathbb{E}{(\mathcal{E}RL_m)}= RL_m+\frac{|\mathcal{X}_m|}{2N}$$

 $$\mathbb{E}{(\mathcal{E}RES_m)}= RES_m +\frac{|\mathcal{X}_m|-1}{2N}$$

 $$\mathbb{E}{(\mathcal{E}PS_m)}= PS_m -\frac{|\mathcal{X}_m|}{2N}.$$

This implies that while the overall empirical $\mathcal{E}CE_m$ may still be a good approximation of the true $\mathcal{P}CE$, caution is warranted when interpreting the individual components. In particular, evaluations of resolution, reliability, and potential score become less reliable for methods that produce a large number of distinct outputs.

\end{document}